%% file: NonparametricForecastingUsingLightCones_main_text.tex
\documentclass[12pt]{article}
\usepackage[authoryear,round]{natbib}
\usepackage{multibib}
\newcites{main}{References}
\newcites{app}{Supplementary References}

\usepackage{amssymb, amsmath}
\usepackage{graphicx}
\usepackage{amsthm}
\usepackage{verbatim}
\usepackage{subfig}
\usepackage{algorithm}
\usepackage[noend]{algorithmic}
\usepackage{hyperref}
\usepackage{setspace}

\newcommand{\gridSep}{3}
\newcommand{\gridSepCm}{3cm}
\newcommand{\opacBackground}{.2}
\newcommand{\opacLC}{0.9}
\newcommand{\opacRectBoundary}{0.7}
\newcommand{\gridColor}{gray}
\newcommand{\opacGrid}{0.5}

\usepackage{geometry}

\usepackage{tikz}
\usetikzlibrary{shapes,arrows}
\usetikzlibrary{positioning}

\input{new_commands}
\newcommand{\asNtoinfty}{\xrightarrow[N\rightarrow\infty]{}}

\newcommand{\ThesisTitle}{LICORS: Light Cone Reconstruction of States for Non-parametric Forecasting of Spatio-Temporal Systems}

\title{LICORS: Light Cone Reconstruction of States for Non-parametric Forecasting of Spatio-Temporal Systems}
\date{\today}

\author{Georg M.\ Goerg \and Cosma Rohilla Shalizi\thanks{Department of
    Statistics, Carnegie Mellon University, Pittsburgh, PA 15213 USA;
    \texttt{\{ gmg, cshalizi \} @ stat.cmu.edu}.  This work was partially
    supported by grants from INET, from the NIH (\# 2 R01 NS047493), and from
    the NSF (DMS1207759).  The authors thank Stacey Ackerman-Alexeeff, Dave
    Albers, Chris Genovese, Rob Haslinger, Martin Nilsson Jacobi, Heike
    J{\"a}nicke, Kristina Klinkner, Cristopher Moore, Jean-Baptiste Rouquier,
    Chad Schafer, Rafael Stern, and Chris Wiggins for valuable discussion, and
    Larry Wasserman for detailed suggestions that have improved all aspects of
    this work.}}

\graphicspath{{}}
\DeclareGraphicsExtensions{.pdf}
 
\begin{document}

\maketitle
\begin{abstract}
  We present a new, non-parametric forecasting method for data where continuous
  values are observed discretely in space and time.  Our method, {\em
    light-cone reconstruction of states} (LICORS), uses physical principles to
  identify predictive states which are local properties of the system, both in
  space and time.  LICORS discovers the number of predictive states and their
  predictive distributions automatically, and consistently, under mild
  assumptions on the data source.  We provide an algorithm to implement our
  method, along with a cross-validation scheme to pick control settings.
  Simulations show that CV-tuned LICORS outperforms standard methods in
  forecasting challenging spatio-temporal dynamics.  Our work provides applied
  researchers with a new, highly automatic method to analyze and forecast
  spatio-temporal data.
\end{abstract}


\section{Introduction}
\label{sec:Introduction}

Many important scientific and data-analytic problems involve fields which vary
over both space and time, e.g., data from functional magnetic resonance
imaging, meteorological observations, or experimental studies in physics and
chemistry.  An outstanding objective in studying such data is prediction, where
we want to describe the field in the future.

Spatio-temporal data being increasingly easy to acquire, manipulate and
visualize, statisticians have developed corresponding methods for statistical
inference, reviewed in works like
\citetmain{Finkenstadt-Held-Isham-spatio-temporal-systems,Cressie-Wikle-spatio-temporal-data}. The
usual tools are a combination of ways of describing the distribution of the
random field (e.g., various dependency measures), and stochastic modeling,
focusing primarily on parametric inference, and secondarily on
parameter-conditional predictions.

While these approaches are valuable, there is a complementary role for direct,
non-parametric prediction of spatio-temporal data, just as with time series
\citepmain{Bosq-nonparametric,Fan-Yao-time-series}. Our aim here is to blend
modern methods of non-parametric prediction with insights from nonlinear
physics on the organization of spatial dynamics, yielding predictors of
spatio-temporal evolution that are computationally efficient and make minimal
assumptions on the data source, but are still accurate and even interpretable.

The idea behind our approach is simply that it takes time for influences to
propagate across space, so we can constrain the search for predictors to a
spatio-temporally local neighborhood at each point.  We combine this with a
novel form of non-parametric smoothing, which infers the prediction (regression
or conditional probability) function by averaging together similar
observations, where ``similarity'' is defined in terms of predictive
consequences, effectively replacing the original geometry of the predictor
variables with a new one, optimized for forecasting.  The combination of these
two tools lets us discover underlying structures, as well make fast and
accurate predictions.

Section \ref{sec:light_cones_and_causal_states} formally defines our prediction
problem and introduces our non-parametric localized approach. Section
\ref{sec:estimating_causal_states} gives the statistical methods to estimate
these optimal predictors from discretely-observed continuous-valued fields.
Section \ref{sec:properties} shows, under weak conditions on the
data-generating process, that our method consistently estimates the predictive
distributions.  Section \ref{sec:simulations} proposes a cross-validation
scheme to choose our control settings, and compares our predictive accuracy to
standard time series techniques.  Finally, Section \ref{sec:discussion}
summarizes this new methodology and discusses future work.  Proofs and 
implementation details can be found in the Supplementary Material.

\section{Local Prediction of Spatio-temporal Fields}
\label{sec:light_cones_and_causal_states}

\subsection{Setting and Notation; Light Cones}

We observe a random field $(\field{X}{r}{t})_{\vecS{r}\in\mathbf{S}, t \in \T}$
in discrete space and time.  The field takes values in a set $\mathcal{X}$,
which may be discrete or continuous.  Space $\mathbf{S}$ is a regular lattice,
equipped with norm $\vnorm{ \vecS{r} }$.  Time $\T$ is taken to be the positive
integers up to $T$.

Suppose that disturbances or influences in the system have a maximum speed of
propagation, $c$. Then the only events which could affect what happens at a
given $\field{}{r}{t}$ are those where $s \leq t$ and $\vnorm{
  \vecS{r}-\vecS{u} } \leq c(t-s)$.  Since this set grows as $s$ recedes into
the past, we call this the {\bf past light cone} (PLC) of $\field{}{r}{t}$. The
{\bf future light cone} (FLC) are all events which could be affected by the
present moment $\field{}{r}{t}$; it thus consists of all those
$\field{}{u}{s}$, where $s > t$ and $\vnorm{ \vecS{r}-\vecS{u} } \leq c(s-t)$.
Light cones look like triangles in $(1+1)D$ fields, and in $(2+1)D$, pyramids (Fig.\
\ref{fig:light_cones}).  Denote the configuration in the past cone of $\field{}{r}{t}$ by
$\field{L^{-}}{r}{t}$:
\begin{equation}
  \label{eq:PLC}
  \field{L^{-}}{r}{t} = \left\{ \field{X}{u}{s} \mid s \leq t, \vnorm{ \vecS{r}-\vecS{u} } \leq c(t-s)\right\}
\end{equation}
$\field{L^{+}}{r}{t}$ is, similarly, the configuration in the future cone.

The spatio-temporal prediction problem is thus: use the configuration of the
past cone, $\field{L^{-}}{r}{t}$, to forecast the configuration of the future
cone, $\field{L^{+}}{r}{t}$.  Light-cone prediction compromises between
capturing global patterns and needing only local information.  We will
construct optimal predictors for light cones presently.  Light cones can be
defined for spatial extended patches of points.  (When the ``patch'' becomes
the whole spatial lattice, we are back to global prediction.)  This leads to a
parallel theory of prediction, but it turns out that the predictive state of a
patch is determined by the predictive states of its points \citepmain[\S 3.3,
Lemma 2 and Theorem 3]{CRS-prediction-on-networks}, so we lose no information,
and gain tractability, by not considering cones for patches.

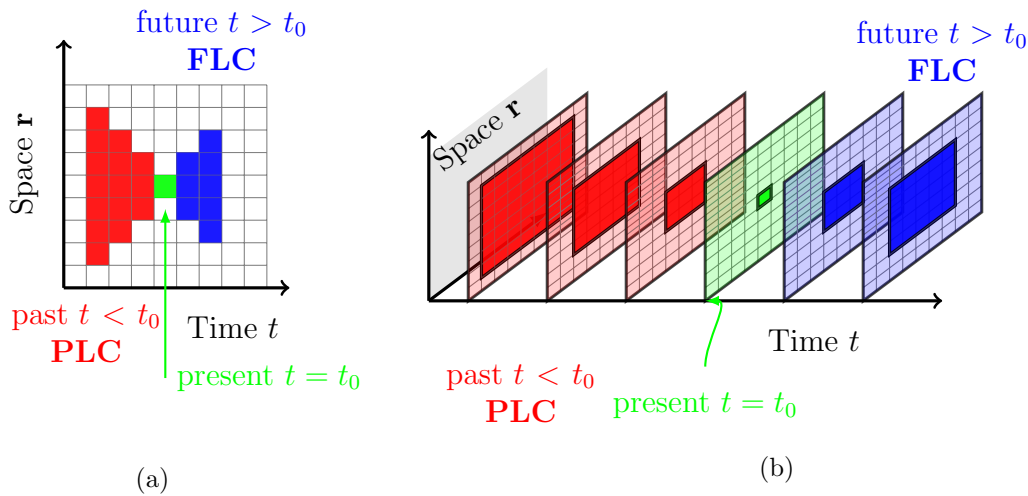
\begin{figure}[!t]
\subfloat[][]{\label{fig:2D_light_cone_tikz_input} 
\begin{minipage}{0.3\textwidth}
\input{2D_light_cone_tikz_input}
\end{minipage}
}%
\hspace{0.05\textwidth}
\subfloat[][]{\label{fig:3D_light_cone_tikz_input}
\begin{minipage}{0.65\textwidth}
\input{3D_light_cone_tikz_input}
\end{minipage}
}%
\caption{\label{fig:light_cones} Past (red) and future (blue) light cones in a
  $(1+1)D$ (a) and $(2+1)D$ (b) system.  Here $c$, the velocity of signal propagation, is set to
  $1$.  The past cone is truncated at a horizon of $h_p=3$ steps, while the
  future cone's horizon is only $h_f=2$.  Whether the present (green) is
  included in the past or the future cone is a matter of convention; see
  Section \ref{sec:simulations}.}
\end{figure}

Computationally, we need to truncate the cones at a finite number of time steps
--- we will call these the {\em past horizon} $h_p$ of $L^-$, and likewise the
{\em future horizon} $h_f$ of $L^+$.  Doing this reduces $L^+$ and $L^-$ to
finite-dimensional random vectors.  (For instance, in Fig.\
\ref{fig:light_cones}, with $h_p=3$ and $c=1$, $\field{\ell^{-}}{r}{t}$ has 15
degrees of freedom.)  The horizons are control settings, and may be tuned
through (for example) cross-validation (\S \ref{sec:CV}).  Similarly, when the
maximum speed of propagation $c$ is not given from background knowledge, it is
also a control setting.

\subsection{Predictive States}

To predict the future $\field{L^{+}}{r}{t}$ from a particular past
configuration, say $\ell^{-}$, requires knowing the conditional distribution
\begin{equation}
  \label{eq:conditional_predictive_distribution}
  \Prob{\field{L^{+}}{r}{t} \mid \field{L^{-}}{r}{t} = \ell^{-}}
\end{equation}
for all $\ell^{-}$.  (Subsequently $\field{}{r}{t}$ may be omitted for
readability.)  Since treating this conditional distribution as an arbitrary
function of $\ell^-$ is not feasible statistically or computationally, we try
to find a {\em sufficient statistic} $\eta$ of past configurations that keeps
the predictive information:
\begin{equation}
  \Prob{ \field{L^{+}}{r}{t} \mid \field{H}{r}{t} = \eta(\ell^{-}) } =  \Prob{ \field{L^{+}}{r}{t} \mid \field{L^{-}}{r}{t} = \ell^{-}} ~.
\end{equation}
There are usually many sufficient statistics $\eta, \eta^{\prime},\ldots $.
When $\eta$ and $\eta^{\prime}$ are both sufficient, but $\eta(\ell^{-}) =
f(\eta^{\prime}(\ell^{-}))$ for some $f$, then $\eta$ is a smaller, more
compressed, summary of the data than $\eta^{\prime}$, and so the former is
preferred by Occam's Razor.  The minimal sufficient statistic $\epsilon$
compresses the data as much as can be done without losing any predictive power,
retaining only what is needed for optimal predictions.

We now construct the minimal sufficient statistic, following
\citetmain{CRS-prediction-on-networks}, to which we refer for some mathematical
details.

\begin{definition}[Equivalent configurations]
  \label{def:equivalent_configurations}
  The past configurations $\ell_i^{-}$ at $\field{}{r}{t}$ and $\ell_j^{-}$ at
  $\field{}{u}{s}$ are {\em predictively equivalent}, $(\ell_i^{-},
  \field{}{r}{t}) \sim (\ell_j^{-}, \field{u}{s})$, if they predict the same
  future with equal probabilities, i.e.\ if
  \begin{equation}
    \Prob{\field{L^{+}}{r}{t} \mid \field{L^{-}}{r}{t} = \ell_i^{-} } = \Prob{ \field{L^{+}}{u}{s}  \mid \field{L^{-}}{u}{s} = \ell_j^{-} } 
  \end{equation}
\end{definition}

Let $\left[ (\ell^{-}, \field{}{r}{t}) \right]$ be the equivalence class of
$(\ell^{-}, \field{}{r}{t})$, i.e., the set of all past configurations and
coordinates that predict the same future as $\ell^{-}$ does at
$\field{}{r}{t}$.  Let
\begin{equation}
  \label{eq:epsilon}
  \epsilon(\ell^{-},\field{}{r}{t}) \equiv \left[ \ell^{-} \right]
\end{equation} 
be the function mapping each $(\ell^{-}, \field{}{r}{t})$ to its predictive
equivalence class.  The values $\epsilon$ can take are the {\em predictive
  states}; they are the minimal statistics which are sufficient for predicting
$L^+$ from $L^-$ \citepmain{CRS-prediction-on-networks}.

Since each predictive state has a unique predictive distribution and vice
versa.  We will thus slightly abuse notation to denote by $\mathcal{E}$ both
the set of equivalence classes and the set of predictive distributions, whose
elements we will write $\epsilon_j$.  We will further abuse notation by writing
the mapping from past cone configurations to predictive distributions as
$\epsilon(\cdot)$, leading to the measure-valued random field
\begin{equation}
  \label{eq:measure_valued_field}
  \field{S}{r}{t} := \epsilon\left( \field{L^{-}}{r}{t} \right) ~.
\end{equation}
One can show \citepmain{CRS-prediction-on-networks} that $\field{S}{r}{t}$ is
Markov even if $\field{X}{r}{t}$ is not.  However, $X$ is not an ordinary
hidden Markov random field, since there is an unusual {\em deterministic}
dependence between transitions in $S$ and the realization of $X$, analogous to
that of a chain with complete connections
\citepmain{Fernandez-Maillard-chains-with-complete-connections}.

To be able to draw useful inferences from a single realization of the process,
we must assume some form of homogeneity or invariance of the conditional
distributions.

\begin{assumption}[Conditional invariance]
  The predictive distribution of a PLC configuration $\ell^{-}$ does not change
  over time or space.  That is, for all $\vecS{r}, t$, all $\vecS{u}, s$, and
  all past light-cone configurations $\ell^{-}$,
  \begin{equation}
    (\ell^{-},\field{}{r}{t}) \sim (\ell^{-},\field{}{u}{s})
  \end{equation}
  \label{ass:st-invariance}
  We may thus regard $\sim$ as an equivalence relation among PLC
  configurations, and $\epsilon$ as a function over $\ell^{-}$ alone.
\end{assumption}

This is just {\em conditional} invariance, like the conditional stationarity
for time series used in
\citetmain{Caires-Ferreira-prediction-of-cond-stat-seq}.  It would be implied
by the field being a Markov random field with homogeneous transitions, or of
course by full stationarity and spatial invariance, but it is weaker.
Assumption \ref{ass:st-invariance} lets us talk about {\em the} predictive
distribution of a PLC configuration, regardless of when or where it was
observed, and to draw inferences by pooling such observations.  If this
assumption fails, we could in principle still learn a different set of
predictive states for each moment of time and/or each point of space (as in
\citetmain{CRS-prediction-on-networks}), but this would need data from multiple
realizations of the same process.

\section{Estimating Predictive States}
\label{sec:estimating_causal_states}

We extend the work of \citetmain{CRS-prediction-on-networks,QSO-in-PRL} to
continuous-valued fields, introducing statistical methods to estimate and
predict non-linear dynamics accurately and efficiently, while still obtaining
insight into the spatio-temporal structure.  Algorithmic details are given in
the Supplementary Material.

Assume we have $T$ consecutive measurements of the field $\field{X}{r}{t}$,
observed over the lattice $\vecS{S}$, with $N = \card{\vecS{S}} \cdot T$
space-time coordinates $\field{}{r}{t}$ in all. Each one of these $N$
point-instants has a past and a future light-cone configuration,
$\field{\ell^{-}}{r}{t}$ and $\field{\ell^{+}}{r}{t}$, represented as,
respectively, $n_p$ and $n_f$ dimensional vectors.  Since predictive states are
sets of PLC configurations with the same predictive distribution, we need to
test this sameness, based on conditional samples $\lbrace \ell^{+} \mid
\ell_i^{-} \rbrace_{i=1}^{N}$ from the observed field.  We will apply
non-parametric two-sample tests for $H_0 ~: ~\Prob{L^{+} \mid L^{-} =
  \ell_i^{j} } = \Prob{L^{+} \mid L^{-} = \ell_i^{-}}$ pairwise for all $i$ and
$j$.  Because there are typically a great many past light cones (one for each
point-instant), and light-cone configurations are themselves high-dimensional
objects, we generally must do this step-wise.

\subsection{Partitioning PLC Configurations: Similar Pasts Have Similar
  Futures}
\label{sec:partition}

It is often reasonable to assume that the mapping from the past to predictive
distributions is regular, so that if two historical configurations are close
(in some suitable metric), then their predictive distributions are also close.
This lets us avoid having to do some pairwise tests, as their results can be
deduced from others.

\begin{assumption}[Continuous histories]
  \label{ass:continuous_dynamics}
  For every $\rho > 0$, there exists a $\delta > 0$ such that 
  \begin{equation}
    \label{eq:cont_assumption}
    \vnorm{ \ell_i^{-} - \ell_j^{-} } < \delta \Rightarrow \KL{\Prob{L^{+} \mid \ell_i^{-}}}{\Prob{L^{+} \mid \ell_j^{-}}} < \rho,
  \end{equation}
  where $\KL{p}{q}$ is the Kullback-Leibler divergence between distributions
  $p$ and $q$ \citepmain{Kullback-info-theory-and-stats}.
\end{assumption}

Assumption \ref{ass:continuous_dynamics} requires that sufficiently small
changes ($< \delta$) in the local past make only negligible ($ < \rho$) changes
to the distribution of local future outcomes.  Statistically, such
smoothness-in-distribution lets us pool observations from highly similar PLC
configurations, enhancing efficiency; physically, it reflects the smoothness of
reasonable dynamical mechanisms.  Chaotic systems, where the exact trajectory
depends sensitively on initial conditions, do not present difficulties, since
Assumption \ref{ass:continuous_dynamics} is about the conditional {\em
  distribution} of the future given partial information on the past, and chaos
has long been recognized as a way to stabilize such distributions, forming the
basis for prediction and control of chaos \citepmain{Kantz-Schreiber-2nd}.

We use Assumption \ref{ass:continuous_dynamics} to justify an initial
``pre-clustering'' of the PLC configuration space, greatly reducing
computational cost with little damage to predictions.  We first divide the PLC
configuration space using fast clustering algorithms into $K \ll N$ clusters,
and then test equality of distributions between clusters ($\mathcal{O}(K^2)$),
rather than light cones ($\mathcal{O}(N^2)$).

When $N$ is small enough, we can skip this initial pre-clustering.  To simplify
exposition, we treat this as assigning each distinct past cone to its own
cluster.

\subsection{Partitioning Clusters into Predictive States}
\label{sec:testing_equality}

Each cluster $P_k$ contains a set of similar PLC configurations, and also
defines a sample of conditional FLCs, $\mathbf{F}_k(\delta) = \lbrace
\ell^{+}_j \mid \ell^{-}_j \in P_k \rbrace \in \R^{N_k \times n_f}$, $k = 1,
\ldots, K$. Since all $\ell_j^{-} \in P_k$ have very similar distribution,
$\mathbf{F}_k \sim Q$ is an approximate sample from the predictive distribution
$\Prob{ \field{L^{+}}{r}{t} \mid \ell^{-} \in P_k}$.  Lemma
\ref{lem:F_i_is_sample_from_p_i}, below, shows that for sufficiently small
$\delta$, $\mathbf{F}_{k_i}(\delta)$ is an exact sample of
$p(\epsilon(P_{k_i}))$.  Thus, to simplify the exposition, we ignore the $\rho$
difference in this section.

Thus, finding equivalent clusters reduces to testing hypotheses of the form
$H_0: p_{k_i} = p_{k_j}$ based on the two samples $\mathbf{F}_{k_i}(\delta)$
and $\mathbf{F}_{k_j}(\delta)$.  For $h_f = 0$ and $c = 1$, FLCs are
one-dimensional and we can use a Kolmogorov-Smirnov test (or any other
two-sample univariate test). In general, however, $\mathbf{F}_k$ are samples
from a very high-dimensional distribution, and we use non-parametric,
multivariate, two-sample tests \citepmain[see
e.g.][]{Rosenbaum05_ExactTwosampleTest,
  RizzoSzekely10_DISCO,Grettonetal07_twosample_kernel}.  Any test satisfying
Assumption \ref{ass:existence_consistent} could be used.

To estimate the predictive states from an initial partitioning of PLC
configurations, we iterate through the list of configurations, recursively
testing equality of distributions.  To initialize the algorithm, create the
first predictive state $\epsilon_1$, containing the first configuration
$\ell_1^{-}$. Then take $\ell^{-}_{2}$ and test if its distribution is equal to
that of $\epsilon_1$.  If it is (at the level $\alpha$), then put
$\ell_{2}^{-}$ in $\epsilon_1$; otherwise generate a new predictive state
$\epsilon_2$ with $\ell_{2}^{-}$.  Then test the next configuration against all
previously established predictive states and proceed as before.  This continues
until all configurations have been assigned to a predictive state.

The predictive distribution of each predictive state can be found by applying
any consistent non-parametric density estimator to the future cone samples
belonging to that state. If we only want point forecasts, we can skip
estimating the whole predictive distribution and just get (e.g.)  the mean of
the samples.

\section{Consistency}
\label{sec:properties}

LICORS consistently recovers the correct assignment of past cone configurations
to predictive states, and the predictive distributions, under weak assumptions
on the data-generating process.  These allow for the number of predictive
states to grow slowly with the sample size, so that we have non-parametric
consistency.  We give all assumptions and lemmas in the main text; proofs are
in the Supplementary Material.

\subsection{Assumptions}
\label{sec:notation_assumptions}

Let $N = \card{\vecS{S} \times \T}$ be the total number of space-time points at
which we observe both the past and future light cone.  We presume that 
$N \rightarrow \infty$, without caring whether $\card{\vecS{S}} \rightarrow \infty$, $\card{\T} \rightarrow \infty$, or both.

% [[CRS: following commented out because they ARE usual so we don't need to
% waste space on the reminders]].
% (We employ the usual asymptotic order symbols: $f(n) =
% \littleO{g(n)}$ when the ratio $f(n) / g(n) \rightarrow 0$ for $n \rightarrow
% \infty$; $f(n) = \bigO{g(n)}$ if $f(n) / g(n) \rightarrow c$ as $n \rightarrow
% \infty$ , with $c$ finite and positive constant; $f(n) = \OmegaO{g(n)}$ if
% $f(n)$ grows at least as fast $g(n)$ --- up to a multiplicative positive
% constant.)

\begin{assumption}[Slowly growing number of predictive states]
\label{ass:number_of_states}
The number of predictive states, $\card{\mathcal{E}} = m(N) = \littleO{N}$, and
always $ \leq N$.
\end{assumption}

Assumption \ref{ass:number_of_states} only guarantees that \emph{at least} one
of the predictive states grows in size.  To bound testing error probabilities,
the number of light cones seen in \emph{every} state must grow as $N$ grows.

\begin{assumption}[Increasing number of light cones in each state]
\label{ass:number_of_samples_in_state}
The number of light cones in each state, $N_j := \card{\epsilon_j}$, grows with
$N$: for all $\epsilon_j \in \mathcal{E}$,
\begin{equation}
\lim_{N \rightarrow \infty}{N_j(N)} = \infty
\end{equation}  
\end{assumption}

Let $N_{\min} = \min_j N_j$ be the number of samples in the smallest predictive
state; thus also $N_{\min} \rightarrow \infty$ for $N \rightarrow
\infty$. Assumption \ref{ass:number_of_samples_in_state} means that the system
re-visits each predictive state as it evolves, i.e., all states are recurrent.
This lets us learn the predictive distribution of each state, from a growing
sample of its behavior.

\begin{assumption}[Bounded conditional distributions] 
\label{ass:bounded_distributions}
All predictive distributions $\epsilon_j \in \mathcal{E}$ have densities with
respect to a common reference measure $\nu$, and $0 < \iota < d\epsilon_j/d\nu <
\kappa <\infty$, for some constants $\iota$ and $\kappa$.
\end{assumption}

This merely technical assumption guarantees bounded likelihood ratios.

\begin{assumption}[Distinguishable predictive states]
\label{ass:distance_between_states}
The KL divergence between states is bounded from below: $\forall i \neq j$,
\begin{equation}
\label{eq:bounded_KL}
0 < d_{\min} \leq \KL{\epsilon_i}{\epsilon_j} =: d_{i,j}
\end{equation}
\end{assumption}

We do not need $d_{i,j} < \infty$.  (In fact, $\KL{\epsilon_i}{\epsilon_j} =
\infty$ is helpful.)  \eqref{eq:bounded_KL} is automatically satisfied for any
fixed number of states.  For an increasing state space, $m = m(N)$, assume
\begin{equation}
\label{eq:inf_d_geq_0}
\inf_{i, j \in m(N)} d_{i,j} = d_{\min} > 0 \text { for } N \rightarrow \infty.
\end{equation}

\begin{lemma}[Conditionally independent FLCs]
  \label{lem:conditional_independent_FLC_given_states}
  If the cones
  $\field{L^+}{r}{t}$ and $\field{L^+}{u}{s}$ do not overlap, then
  \begin{equation}
    \label{eq:conditional_independent_FLC_given_states}
    \field{L^+}{r}{t} \indep \field{L^+}{u}{s} \mid \field{S}{r}{t}, \field{S}{u}{s}.
  \end{equation}
  In particular,
  \begin{equation}
    \label{eq:conditional_independent_FLC_given_states_distributions}
    \Prob{\field{L^+}{r}{t}, \field{L^+}{u}{s} \mid \field{S}{r}{t}, \field{S}{u}{s}} = \Prob{\field{L^+}{r}{t}  \mid \field{S}{r}{t}} \Prob{\field{L^+}{u}{s} \mid  \field{S}{u}{s}} ~.
  \end{equation}
\end{lemma}

\begin{corollary}
  \label{cor:univariate_FLCs_conditional_independent}
  If $h_f = 0$, then FLCs are conditionally independent given their predictive
  state.
\end{corollary}

\subsubsection{Getting samples from $\epsilon_i$} 
\label{sec:getting_samples}

We get a sample of FLCs from the predictive distribution of $\ell_i$ by first
taking all PLCs in a $\delta$-neighborhood around $\ell_i$,
\begin{equation}
  \label{eq:indices_PLC_neighborhood}
  I_i(\delta) = \lbrace j \mid \vnorm{ \ell_i^{-} - \ell_j^{-} } < \delta \rbrace.
\end{equation}
For later use, we denote by $S_i(N, \delta) = \card{I_i(\delta)}$ the number of
such light cones.  By Assumption \ref{ass:continuous_dynamics}, we get our
sample from $\epsilon_i$ by collecting the corresponding future cone
configurations:
\begin{equation}
  \label{eq:samples_F_i}
  \mathbf{F}_i(\delta) = \lbrace \ell_j^{+} \mid j \in I_i(\delta) \rbrace ,
\end{equation}

\begin{lemma}
  \label{lem:F_i_is_sample_from_p_i}
  For sufficiently small $\delta > 0$, all past configurations in $I_i(\delta)$
  are predictively equivalent: $\forall j, k \in I_i(\delta), \ell^-_j \sim
  \ell^-_k ~$.  Consequently, all $\ell^+_j$, $j \in I_i(\delta)$, are drawn
  from the same distribution $\epsilon(\ell^-_i)$.
\end{lemma}

For finite $N$, it may not be possible in practice to find and use a
sufficiently small $\delta$.  With pre-clustering, for instance, some of the
clusters may have diameters greater than the $\delta$ which guarantees equality
of distribution.  Then the samples $\mathbf{F}_i(\delta)$ are actually from
multiple states.  One could circumvent this by using more clusters, which
generally shrinks cluster diameters, but this would also reduce the number of
samples per neighborhood, increasing the error rate of our two-sample tests.
In practice, then, one must trade off decreasing $\delta$ to discover all
predictive states and keeping a low testing error.

\begin{corollary}
  \label{cor:likelihood_factorizes}
  For sufficiently small $\delta > 0$, and non-overlapping FLCs, all the future
  configurations in $\mathbf{F}_i(\delta)$ are IID samples from
  $\epsilon(\ell^{-}_i)$.
\end{corollary}

In general, for $h_f > 0$ the FLCs in $\mathbf{F}_i(\delta)$ can be overlapping
and the conditional likelihood does not factorize. Yet, without loss of
generality, we can consider only non-overlapping FLCs.  This is because we can
explicitly exclude overlapping FLCs from $\mathbf{F}_i(\delta)$, at the cost of
reducing the sample size to $\tilde{S}_i(N, \delta) \leq S_i(N, \delta)$.  For
each $\ell_i$, the maximum number of FLCs which we must thereby exclude, say
$w$, is fixed geometrically, by $c$, $h_f$ and the dimension of the space
$\vecS{S}$, and does not grow with $N$.  The exclusion thus is asymptotically
irrelevant, since $\frac{S_i(N, \delta)}{w} \leq \tilde{S}_i(N, \delta) \leq
S_i(N, \delta)$.

Further, note that, at least formally, it's enough to analyze the univariate,
zero-horizon FLC distributions, which rules out overlaps.  This is because
longer-horizon FLC distributions must be consistent with the one-step ahead
distributions and the transition relations of the underlying predictive states.
Thus we could get the $n_f$-dimensional FLC distribution by iteratively
combining the univariate FLC distributions and the predictive state
transitions, i.e., by chaining together one-step-ahead predictions, as in
\citetmain[Corollary 2]{CMPPSS}.

\begin{assumption}[Number of samples from each cone]
  \label{ass:number_of_samples_per_LC}
  For each fixed $\delta > 0$, and each past light cone $\ell_i$, $S_i(N,
  \delta) \asNtoinfty \infty$.
\end{assumption}
For each $\delta$, $S_i(N, \delta)$ is a random variable, and to establish
consistency we need some regularity conditions on how $S_i$ grows with $N$.
Let $S_{\min}(N, \delta) = \min_{j} S_{j}(N, \delta)$ be the smallest number of
samples per $\delta$-neighborhood for each $N$ and $\delta$.

\begin{assumption}
  \label{ass:mgf_S_min}
  For some $\tilde{c} > 0$,
  \begin{equation}
    N \cdot m(N) \cdot \E e^{-\tilde{c} d^2_{\min} S_{\min}(N, \delta)} \asNtoinfty 0 ~.
  \end{equation}
\end{assumption}
Since $\E e^{t S_{\min}(N, \delta)}$ is the moment generating function of
$S_{\min}$, this amounts to asserting that the number of samples concentrates
around its mean while growing, ruling out pathological cases where $S_i(N,
\delta)$ grows to infinity, but concentrates around small values.

\subsection{Unknown Predictive States: Two-sample Problem}
\label{sec:two_sample_problem}

With a finite number $N$ of observations, recovering the states is the same as
determining which past cone configurations are predictively equivalent.  We
represent this with an $N\times N$ binary matrix $\mathbf{A}$, where $A_{ij} =
1$ if and only if $\ell_i \sim \ell_j$.  LICORS gives us an
estimate of this matrix, $\widehat{\mathbf{A}}$, and we will say that
the predictive states can be recovered consistently when
\begin{equation}
  \Prob{\widehat{\mathbf{A}} \neq \mathbf{A}} \asNtoinfty 0 ~.
  \label{eqn:consistent-recovery-of-states}
\end{equation}

Since the predictive distributions are unknown, we use non-parametric
two-sample tests to determine whether two past cone configurations are
predictively equivalent.  While simulations can always be used to approximate
the power of particular tests against particular alternatives, there do not
(yet) seem to be any general expressions for the power of such tests, analogous
to the bounds on likelihood tests in terms of KL divergence
\citepmain{Kullback-info-theory-and-stats}.  Nonetheless, we expect that for $N
\rightarrow \infty$, the probability of error approaches zero, as long as the
true distributions are far enough apart.  We thus make the following
assumption.

\begin{assumption}
  \label{ass:existence_consistent}
  Suppose we have $n$ samples from distribution $p$, and $n^{\prime}$ samples
  from distribution $q$, all IID.  Then there exist a positive constants
  $d_{n,n^{\prime}}$ tending to $0$ as $n,n^{\prime} \rightarrow \infty$, and a
  sequence of tests $T_{n,n^{\prime}}$ of $H_0: p=q$ vs.\ $H_1: p \neq q$ with
  size $\alpha = \littleO{\min(n,n^{\prime})^{-2}}$, and type II error rate
  $\beta(\alpha, n,n^{\prime}) = \littleO{\min(n,n^{\prime})^{-2}}$ so long as
  $p$ and $q$ are mutually absolutely continuous and $\KL{p}{q} \geq
  d_{n,n^{\prime}}$.
\end{assumption}

Note that if the number of predictive states is constant in $N$, we can weaken
the assumption to just a sequence of tests whose type I and type II error
probabilities both go to zero supra-quadratically when $\KL{p}{q} \geq
d_{\min}$.

\begin{theorem}[Consistent predictive state estimation]
  \label{thm:consistent_unknown}
  Under Assumptions \ref{ass:st-invariance}, \ref{ass:continuous_dynamics}, 
  \ref{ass:number_of_states}, \ref{ass:number_of_samples_in_state},
  \ref{ass:bounded_distributions}, \ref{ass:distance_between_states},
  \ref{ass:number_of_samples_per_LC}, \ref{ass:mgf_S_min}, and
  \ref{ass:existence_consistent},
  \begin{equation}
    \Prob {\widehat{\mathbf{A}} \neq \mathbf{A} } \asNtoinfty 0.
  \end{equation}
\end{theorem}

\section{Simulations}
\label{sec:simulations}

To evaluate the non-asymptotic predictive ability of LICORS, and to compare it
to more conventional methods, we use the following simulation, designed to be
challenging, but not impossible.  $\field{X}{r}{t}$ is a continuous-valued
field in $(1+1)D$, with a discrete latent state $\field{d}{r}{t}$.  We use
``wrap-around'' boundary conditions, so sites $0$ and $\card{S}-1$ are
adjacent, and the one spatial dimension is a torus.  The observable field
$\field{X}{r}{t}$ is conditionally Gaussian,
\begin{equation}
  \label{eq:CA_cont_predictive_rule_normal}
  \field{X}{r}{t} \mid \field{d}{r}{t}  \sim \begin{cases}
    \mathcal{N}(\field{d}{r}{t}, 1),& \text{ if } | \field{d}{r}{t} | < 4, \\
    \mathcal{N}(0, 1), & \text {otherwise},
  \end{cases}
\end{equation}
with initial conditions $\field{X}{\cdot}{1} = \field{X}{\cdot}{2} =
\mathbf{0}$.  The state space $\field{d}{r}{t}$ evolves with the observable
field,
\begin{equation}
  \label{eq:state_description}
  \small \field{d}{r}{t} = \left[ \frac{\sum_{i=-2}^{2}{\field{X}{r+i \bmod \card{S}}{t-2}}}{5} - \frac{\sum_{i=-1}^{i}{\field{X}{r+i \bmod \card{S}}{t-1}}}{3} \right],
\end{equation}
where $[x]$ is the closest integer to $x$.  In words, Eq.\
\eqref{eq:state_description} says that the latent state $\field{d}{r}{t}$ is
the rounded difference between the sample average of the $5$ nearest sites at
$t-2$ and the sample average of the $3$ nearest sites at $t-1$.  Thus $h_p = 2$
and $c=1$.

If we include the present in the FLC, \eqref{eq:CA_cont_predictive_rule_normal}
gives $h_f = 0$, making FLC distributions one-dimensional and letting us use
the Kolmogorov-Smirnov test.  As $\field{d}{r}{t}$ is integer-valued, a little
calculation shows there are $7$ predictive states, which we may label with
their conditional means as $\lbrace \epsilon_{-3}, \epsilon_{-2}, \ldots,
\epsilon_{2}, \epsilon_{3} \rbrace$. Thus $\field{X}{r}{t} \mid \epsilon_k \sim
\mathcal{N}(k, 1)$.

\begin{figure}[!t]
  \centering \subfloat[]
  {
  \includegraphics[width=.45\textwidth]{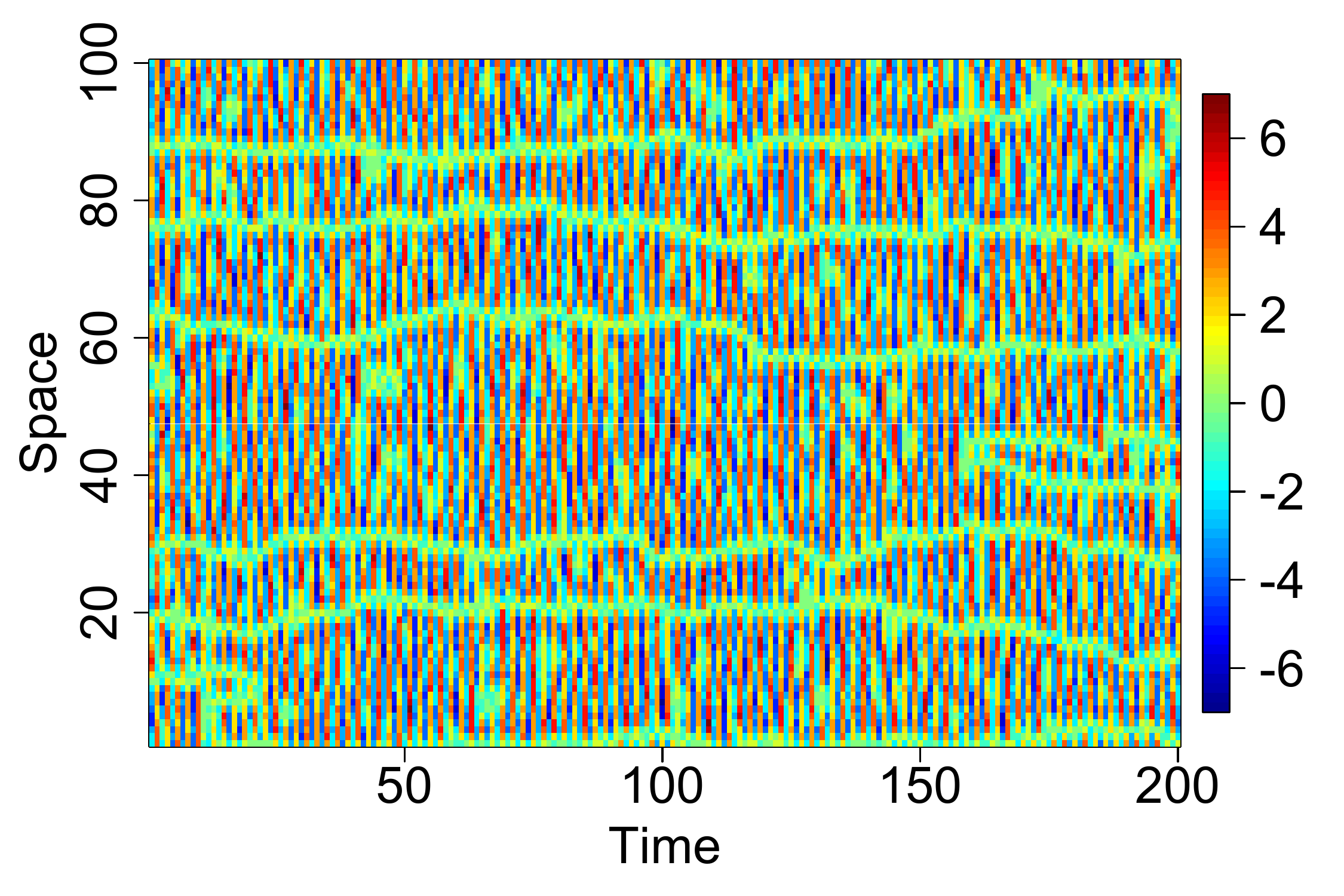}\label{fig:cont_CA_round_normal_states}
  } 
  \hspace{0.01\textwidth} 
  \subfloat[]
  {
  \includegraphics[width=.45\textwidth]{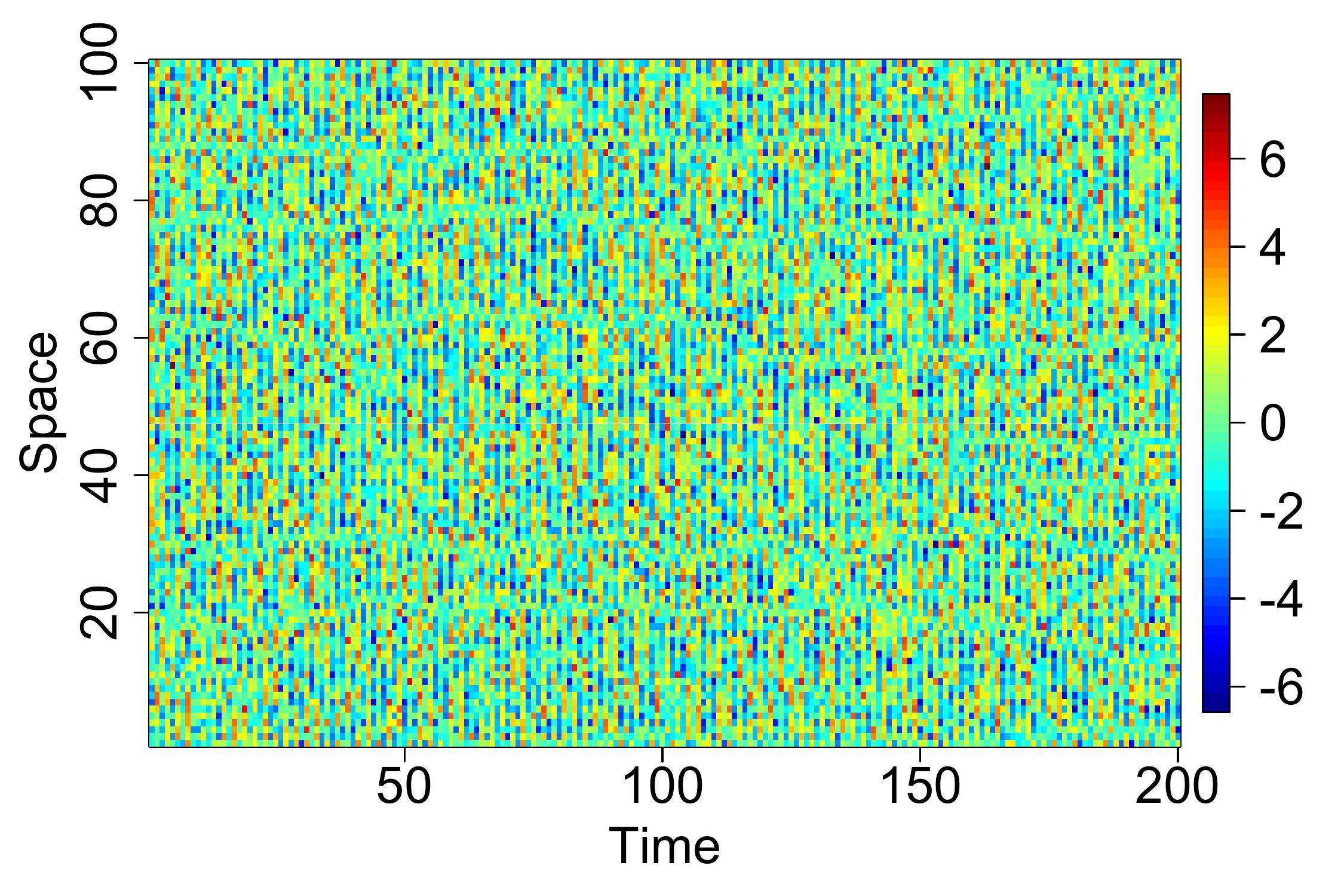}\label{fig:cont_CA_round_normal_observed}
  }

  \caption{\label{fig:cont_CA_round_states_Simulations} Simulation of
    \eqref{eq:CA_cont_predictive_rule_normal}--\eqref{eq:state_description}:
    (a) state-space $\field{d}{r}{t}$, (b) observed field $\field{X}{r}{t}$.
    Space (100 cells) runs vertically, time (200 steps, first 100 discarded for
    burn-in) runs from left to right.}
\end{figure} 

\begin{figure}[!t]
  \centering 
  \subfloat[]
  {
  \includegraphics[width=.45\textwidth]{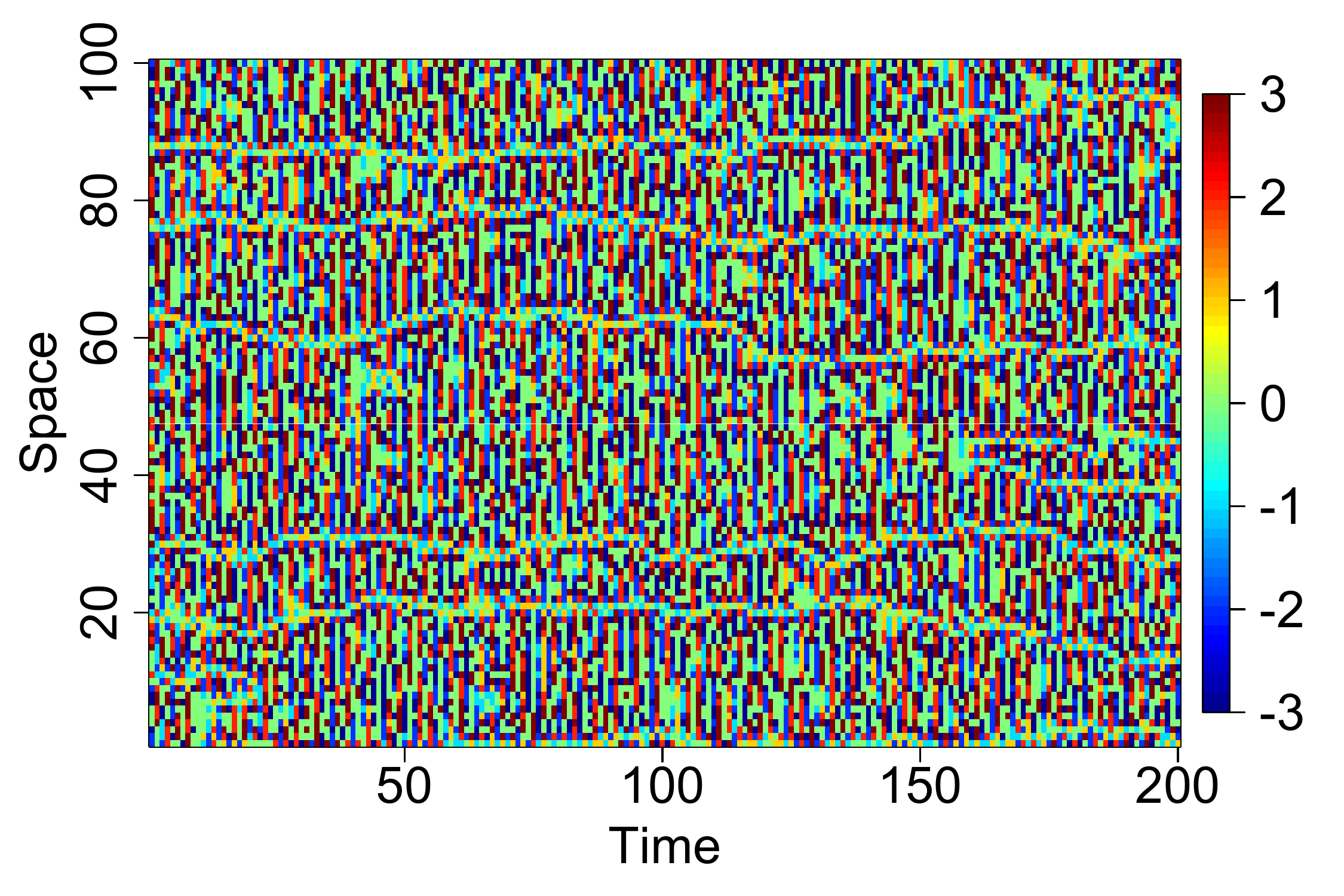}\label{fig:cont_CA_round_normal_predictive_states}
  }
  \hspace{0.02\textwidth}
  \subfloat[]
  {
  \includegraphics[width=.45\textwidth]{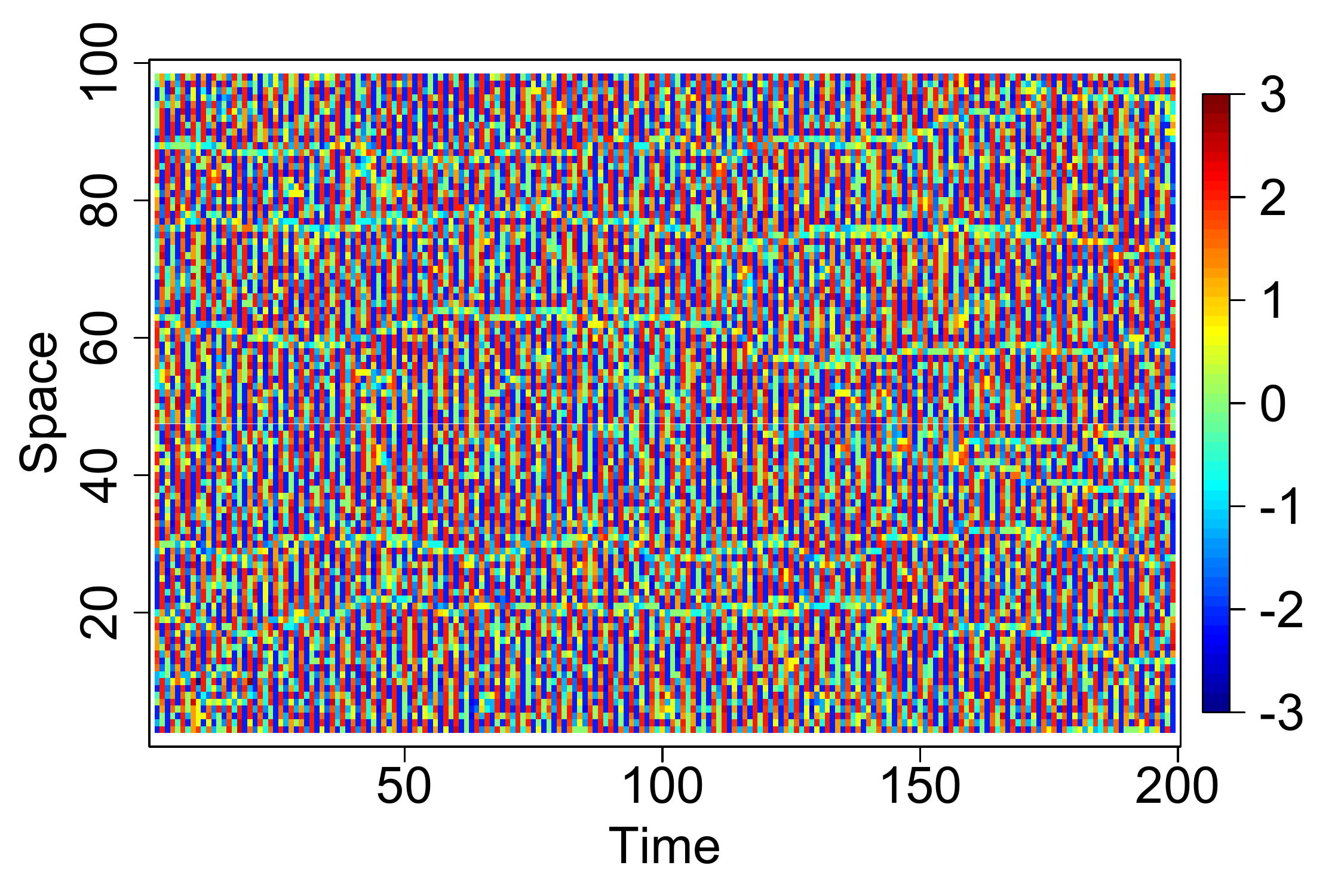}\label{fig:cont_CA_round_normal_predictions}
  }
  \caption{\label{fig:cont_CA_round_states_LICORS} Comparison of true and
    estimated predictive distributions.  (a) true predictive state
    $\field{S}{r}{t}$, with points colored by conditional expectations; (b)
    LICORS predictions, with states and distributions reconstructed using
    $k=50$ nearest neighbors (fixed) and $h_p=2, \alpha=0.2$ (chosen by
    cross-validation).}
\end{figure} 

Figure \ref{fig:cont_CA_round_states_Simulations} shows one realization of
\eqref{eq:CA_cont_predictive_rule_normal}--\eqref{eq:state_description}.  The
latent states have clear spatial structures, which is obscured in the observed
field.  Figure \ref{fig:cont_CA_round_normal_predictive_states} shows the true
predictive state space $\field{S}{r}{t}$ (expected value at each at each
$\field{}{r}{t}$); the LICORS estimate $\field{\widehat{S}}{r}{t}$ is shown in
Fig.\ \ref{fig:cont_CA_round_normal_predictions}.  LICORS not only accurately
estimates $\field{S}{r}{t}$, but also learns the prediction rule
\eqref{eq:CA_cont_predictive_rule_normal} from the observed field
$\field{X}{r}{t}$.

\subsection{Forecasting Competition: AR, VAR, and LICORS}
\label{sec:forecasting_comparison}

A brute-force approach to spatio-temporal prediction would treat the whole
spatial configuration at any one time as a single high-dimensional vector, and
then use ordinary, parametric time-series methods such as vector
auto-regressions (VAR) \citepmain{Luetkepohl07}, or non- or semi-
non-parametric models \citepmain{Bosq-nonparametric,Fan-Yao-time-series}.  Such
global approaches suffer under the curse of dimensionality: real data sets may
contain millions of space-time points, so fitting global models becomes
impractical, even with strong regularization \citepmain{Bosq-Blanke-large-dim}.
Moreover, such global models will not be good representations of complex
spatial dynamics.

On the other hand, space can be broken up into small patches (in the limit,
single points), and then one can fit standard time series models to each
patch's low-dimensional time series.  Such local strategies (partially) lift
the curse of dimensionality, and thus make VAR or non-parametric time-series
prediction practical, but creates the problem of selecting good sizes and
shapes for these patches, and ignores spatial dependence across patches.

To show how LICORS escapes this dilemma, we compare it to other forecasting
techniques in a simulation.  Using $100$ replications of
\eqref{eq:CA_cont_predictive_rule_normal} -- \eqref{eq:state_description}, with
$n = 100$ points in space, and $T = 200$ steps in time, we compared LICORS,
with and without pre-clustering, to (a) the empirical time-average of each
spatial point; (b) a separate, univariate $AR(p)$ model for each point; a (c)
separate $VAR(p)$ for each non-overlapping spatial patch of $5$ points; and the
true conditional expectation function.  (See \S \ref{sec:competing-methods} in
the Supplemental Information for details of the competing methods.) 

%%%%%%%%%%%%%
% [[CRS: Did we do a global VAR?  We could have done a global VAR(1), I think.]]
%%%%%%%%%%%%%

\begin{figure}[!t]
  \centering 
  \subfloat[In-sample]
  {
  \includegraphics[width=.48\textwidth]{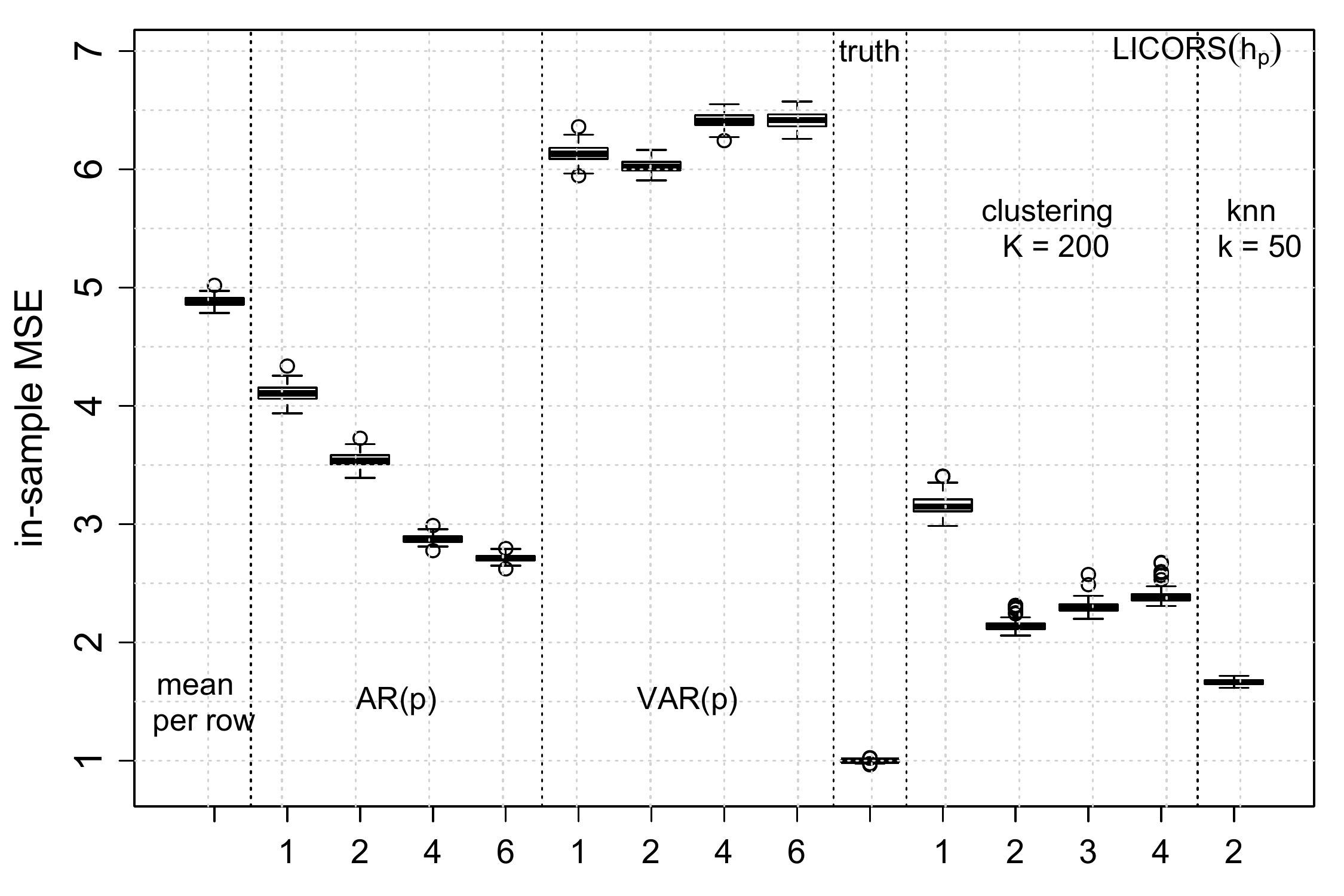}\label{fig:ERM_AR_VAR_LICOR_nclusters_knn}
  }
 % \hspace{0.01\textwidth}
  \subfloat[Out-of-sample]
  {
  \includegraphics[width=.48\textwidth]{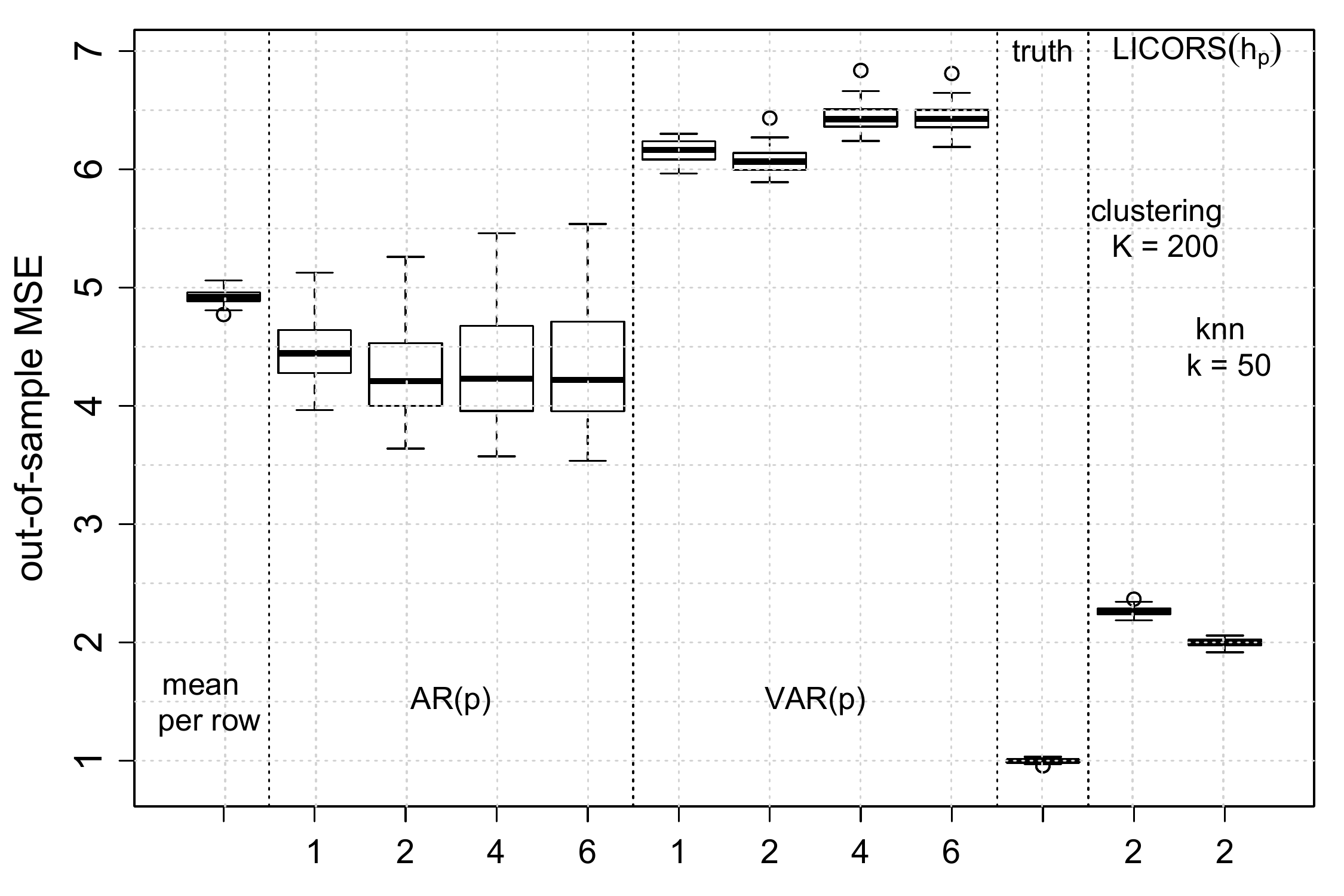}\label{fig:ERM_out_of_sample_AR_VAR_LICOR_nclusters_knn}
  }
  \caption{\label{fig:ERM_in_out_AR_VAR_LICOR_nclusters_knn} MSEs for LICORS
    and parametric competitors on
    \eqref{eq:CA_cont_predictive_rule_normal}--\eqref{eq:state_description}.
    LICORS with pre-clustering used $K=200$ clusters and varying past horizons;
    LICORS without pre-clustering use $k=50$ neighbors and $h_p=2$; both variants
    fixed $\alpha = 0.05$.}
\end{figure}

Figure \ref{fig:ERM_in_out_AR_VAR_LICOR_nclusters_knn} shows for each predictor the
estimated mean squared error (MSE) for the in-sample 
(Fig.\ \ref{fig:ERM_AR_VAR_LICOR_nclusters_knn}) as well as out-of-sample (Fig.\ \ref{fig:ERM_out_of_sample_AR_VAR_LICOR_nclusters_knn}) one-step ahead prediction
error.  Splitting up space while using standard methods appears not to help and
may even hurt.  LICORS performs best among all methods, once $h_p \geq 2$.
While pre-clustering performs worse than direct estimation, it still predicts
much better than the other methods.

Overall, LICORS with $h_p = 2$ gives the best forecasts, where $\alpha = 0.05$
was set in advance.  At no point did we make an assumption about the number of
predictive states or the shape of the conditional distribution.  Even though
the true system is conditionally Gaussian, LICORS out-performed the parametric
Gaussian models.  Thus we expect to do even better on non-Gaussian fields.

Even though we know the true light cone size in simulations, the ``true''
$\alpha$ can not be obtained directly. It controls the number of estimated
predictive states: larger $\alpha$ implies less merging of clusters, and thus
more number of predictive states; smaller $\alpha$ leads to more merging and
hence less states.

In practice, one does not know the true light cone size nor the true number of
states; they are rather control settings which affect the predictive
performance.  As we can accurately measure predictive performance by
out-of-sample MSE, we propose a cross-validation (CV) procedure to tune $h_p$
and $\alpha$.

\subsection{Cross-validation to Choose Optimal Control Settings}
\label{sec:CV}

A good method should learn the invariant predictive structures of the system,
avoiding over-fitting to the accidents of the observed sample.  Ideally, the
method should estimate nearly the same predictive states from (almost) any two
realizations of the same system, while still being sensitive to differences
between distinct systems.  

Cross-validation is the classic way to handle this sensitivity-stability
trade-off, and we use a data-set splitting version of it here.  We simply
divide the data set at its mid-point in time, use its earlier half to find
predictive states, and evaluate the states' performance on the data's later
half; see Supplemental Figure \ref{fig:CV_Overview}.  (Assumption
\ref{ass:st-invariance}, of conditional stationarity, is important here.)
While quite basic, simulations show that it does indeed find good control
settings.

Using the same realizations of the model system as in the forecasting
competition, we tried all combinations of $h_p \in \lbrace 1,2,3 \rbrace$ and
$\alpha \in \lbrace 0.3, 0.2, 0.15, 0.1, 0.05, 0.01, 0.001 \rbrace$.  We picked
the control settings to do well on the continuation of the sample realization,
but since this is a simulation, we can also check that these settings perform
well on an {\em independent} realization of the same process.  Figure
\ref{fig:CV_LICOR_comparisons} compares, for the selected control settings, the
in-sample MSE on the first half of each realization, the MSE on the second
half, and the MSE on all of a completely independent realization, for both the
direct and the pre-clustered versions of LICORS.  (As before, direct estimation
does a bit better than pre-clustering.)  There is little difference between the
MSEs on the continuation of the training data and on independent data,
indicating little over-fitting to accidents of particular sample paths.  (See
\S \ref{sec:excess-risk} in the supplemental information for further details.)
Notably, CV picked the optimal $h_p$, namely 2, on all 100 trials.

\begin{figure}[!t]
  \centering 
  \includegraphics[width=.8\textwidth]{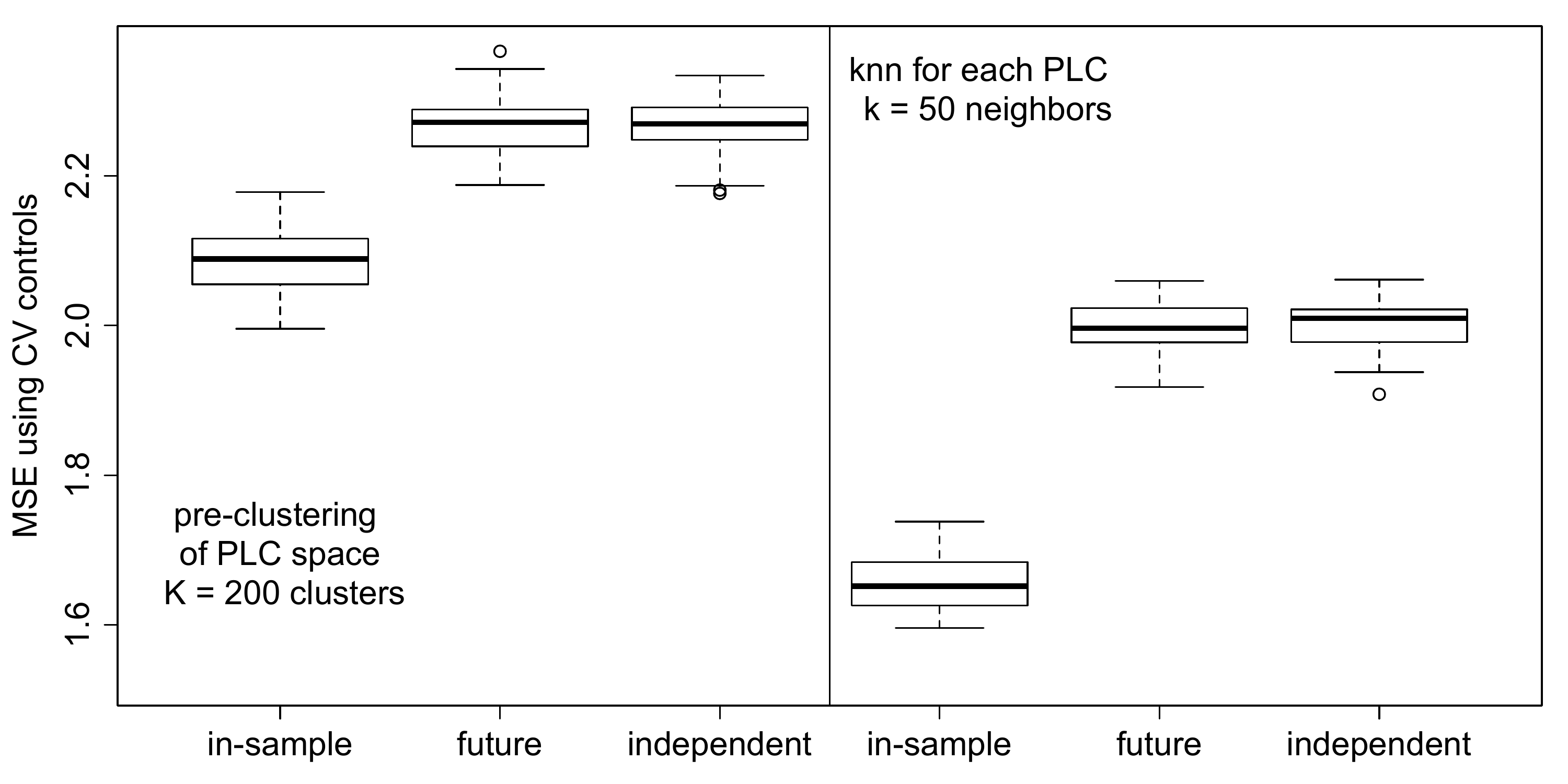}
  \caption{\label{fig:CV_LICOR_comparisons} Cross-validation for LICORS: MSE,
    using the CV-picked control settings, on the first half of each realization
    (``in-sample''), on the second half (``future''), and on all of an
    independent realization (``independent'').}
\end{figure}

%%%%%%%%%%%%%%%%%%%%%%%%%%%%
% CRS: change "parameters" in vertical axis label to "controls"
% remove "CV" from horizontal axis labels ("in-sample CV" is a contradiction
% in terms); as before, the numbers beside the bar plots are much too small and
% hard to interpret, so if we really want them they need to be a table in the
% SI
% GMG: Done.
%%%%%%%%%%%%%%%%%%%%%%%%%%%%%%%%

As expected, the smaller the value of $\alpha$ picked by CV, the more merging
between clusters, and the smaller the number of states (see Supplemental Figure
\ref{fig:CV_Excess_risk_vs_states}).  Here, the true number of states $m=7$,
but both pre-clustering and direct estimation give much higher $\widehat{m}$
(10--30 with pre-clustering, 30--90 without).  The gap appears to be due to
cross-validating pushing (in this context) for lower approximation error and
more states, rather than fewer states and lower estimation error (\S
\ref{sec:excess-risk} in the supplemental information).  Having $\widehat{m}$
be substantially larger than $m$ thus does not degrade out-of-sample
predictions.

\section{Discussion}
\label{sec:discussion}

\subsection{Related Work}
\label{sec:related_work}

Predictive state reconstruction estimates the prediction processes introduced
by \citetmain{Knight-predictive-view}.  Knight's construction is for stochastic
processes $X$ with a single, continuous time index; but since $X_t$ can take
values in infinite-dimensional spaces, most useful spatial models can
implicitly be handled in this way, and by considering discrete time we avoid
many measure-theoretic complications.  After Knight, the same basic
construction of the prediction process was independently rediscovered in
nonlinear dynamics and physics \citepmain{Inferring-stat-compl,CMPPSS}, in
machine learning
\citepmain{Jaeger-operator-models,predictive-representations-of-state,Langford-Salakhutdinov-Zhang},
and in the philosophy of science \citepmain{Salmon-1984}.

Spatio-temporally local prediction processes were introduced in
\citetmain{CRS-prediction-on-networks,QSO-in-PRL} to study self-organization
and system complexity, along lines suggested by
\citetmain{Grassberger-1986,Inferring-stat-compl}.  A related proposal was made
by was made by \citetmain{Parlitz-Merkwirth-local-states}, and light cones have
been used in stochastic models of crystallization
\citepmain{Capasso-Micheletti-spatial-birth-growth}, going back to
\citetmain{Kolmogorov-light-cones}.

While the prediction-process formalism allows for continuous-valued observable
fields, the prior work by Shalizi {\em et al.} only gave procedures for
discrete-valued fields.  J\"anicke {\em et al.}  used those procedures on
continuous-valued data by discretizing them
\citepmain{Jaenicke09_Diss,Janicke-Scheuermann-automatic-visualization,Jaenicke-et-al-multifield-visualization}.
We avoid discretization by using methods to estimate and compare continuous,
high-dimensional distributions.

\subsection{Conclusion}

We introduce a new non-parametric method, LICORS, for spatio-temporal
prediction.  LICORS learns the predictive geometry in the state space
underlying the system, by clustering observations according to the similarity
of their local predictive distributions. Together with our cross-validation
scheme, LICORS is a fully data-driven, non-parametric method to learn and use
the non-linear, high-dimensional dynamics of a large class of spatio-temporal
systems. The good performance of the CV procedure (Fig.\
\ref{fig:CV_Overview}) suggests that using it to find control settings in
applications will avoid over-fitting.

Under weak assumptions, LICORS consistently estimates predictive distributions.
Simulations show that it largely outperforms standard prediction methods.  We
have motivated presented results for $(1+1)D$ fields, but both the theory and
practice extend without modification to higher-dimensional fields.  While it
will be good to extend LICORS to handle continuous predictive states, and to
derive theoretical guarantees about its behavior under cross-validation, it can
already be applied to experimental data.  It provides a powerful, principled
tool for forecasting complex spatio-temporal systems.

\bibliographystylemain{plainnat}
\bibliographymain{../../../bib/PhD_thesis,../../../bib/locusts}

\newpage
\appendix
\setcounter{page}{1}
\input{NonparametricForecastingUsingLightCones_supplementary_material_input}

\end{document}

%% file: new_commands.tex
\newcommand{\E}{\mathbb{E}}

\newcommand{\R}{\mathbb{R}}

\newcommand{\T}{\mathbb{T}}

\newcommand{\Prob}{\mathbb{P}} % GMG version
\renewcommand{\Prob}[1]{\mathbb{P}\left( #1 \right)} % Cosma version

\newcommand{\indep}{\protect\mathpalette{\protect\independenT}{\perp}}
\def\independenT#1#2{\mathrel{\setbox0\hbox{$#1#2$}%
\copy0\kern-\wd0\mkern4mu\box0}}

\newcommand{\vnorm}[1]{\lVert #1 \rVert}

\newcommand{\card}[1]{| #1 |}

\newcommand{\KL}[2]{\mathcal{D}_{KL}\left(#1 \mid \mid #2 \right)} 
\newcommand{\SymKLexp}[2]{\frac{\KL{p}{q} + \KL{q}{p}}{2}}

\newcommand{\vecS}[1]{\mathbf{#1}}
\newcommand{\field}[3]{#1\left(\vecS{#2}, #3 \right) }

\newcommand{\bigO}[1]{\ensuremath{\operatorname{\mathcal{O}}\bigl(#1\bigr)}}
\newcommand{\littleO}[1]{\ensuremath{\operatorname{o}\bigl(#1\bigr)}}

\newcounter{bar}
\newcommand{\foo}{%
        \stepcounter{bar}%
        \thebar}

\newtheorem{theorem}{Theorem}[section]

\newtheorem{assumption}[theorem]{Assumption}

\newtheorem{corollary}[theorem]{Corollary}

\newtheorem{definition}[theorem]{Definition}

\newtheorem{lemma}[theorem]{Lemma}

\newtheorem{properties}[theorem]{Properties}

%% file: 2D_light_cone_tikz_input.tex
% http://www.texample.net/tikz/examples/swan-wave-model/

\begin{tikzpicture}[scale=.3,every node/.style={minimum size=1cm},on grid,rotate=-90]
		
    %slanting: production of a set of n 'laminae' to be piled up. N=number of grids.

    \begin{scope}[
    	yshift=0,every node/.append style={
    	    yslant=0,xslant=0},yslant=0,xslant=0
    	             ]
        \draw[step=1cm, \gridColor] (0,0) grid (9,9);
        \fill[green,fill opacity=\opacLC] (4,4) rectangle (5,5);
        \foreach \x in{0, 1, 2}
            \fill[red, fill opacity=\opacLC] (3-\x,3-\x) rectangle (6+\x,4-\x);
        \foreach \x in{0, 1, 2}
            \draw[step=1cm, \gridColor, opacity = \opacGrid]  (3-\x,3-\x) grid (6+\x,4-\x);
        \foreach \x in{0, 1}
			\fill[blue, fill opacity=\opacLC] (3-\x,5+\x) rectangle (6+\x,6+\x);
		\foreach \x in{0, 1}
            \draw[step=1cm, \gridColor, opacity = \opacGrid] (3-\x,5+\x) grid (6+\x,6+\x);
            
        %\draw[black, very thick, opacity = \opacRectBoundary] (0,0) rectangle (9,9);
    \end{scope}

	%\draw[color=black,->, very thick] (0,0,0) -- (10,0,-2);
	\draw[color=black,->, very thick] (9,0,0) -- node [sloped, below, near end, rotate =-90] {Time $t$}  (9,10,0);
	\draw[color=black,->, very thick] (9,0,0) -- node [sloped, above, rotate = 90] {Space $\mathbf{r}$}  (-2,0,0);
%	\draw[color=black,->, very thick] (4.5,-3.5*\gridSep,0) -- node [sloped, right, near end, rotate = 37] {Space $\mathbf{r}$}  (1,-3*\gridSep,0);

	%\foreach \x in{1, 2, 3}
    %\draw[color=blue,->, thick] (4.5 - 0.5*\x, -3*\gridSep, \x) -- (4.5 - \x, 0, \gridSep*\x);
    
    %\draw[color=red,->, thick] (4, -3*\gridSep, 1) -- (2, 0, 2);
    %\draw[color=red,->, thick] (1, -2*\gridSep, 4) -- (1, 3*\gridSep, 4);

    %putting arrows and labels:
    \draw[-latex,thick, green] (13,4.5) node[right]{present $t=t_0$} to[out=180,in=0] (5.5,4.5);

    %\draw[-latex,thick](5.8,-.3)node[right]{$\mathsf{Comp.\ G.}$}
    %    to[out=180,in=90] (3\opacLC,-1);

    %\draw[-latex,thick] (6,-2) node[right]{past light cone}
    %     to[out=180,in=180] (4,-2);

	\draw (11,1) node [text width=2.5cm, text centered, red]{past $t <t_0$ \\ \textbf{PLC}};
	\draw (-2, 7) node [text width=2.5cm, text centered, blue]{future $t>t_0$ \\ \textbf{FLC}};
    %drawing points on grid's conrners.

\end{tikzpicture}

%% file: 3D_light_cone_tikz_input.tex
% http://www.texample.net/tikz/examples/swan-wave-model/

\begin{tikzpicture}[scale=.35,every node/.style={minimum size=1cm},on grid,rotate=-90]
		
    %slanting: production of a set of n 'laminae' to be piled up. N=number of grids.

    \begin{scope}[
    	yshift=-3.5*\gridSepCm,every node/.append style={
    	    yslant=0,xslant=0},yslant=0,xslant=-0.75]
        %\fill[gray,fill opacity=\opacLC] (1,1) rectangle (4.5,4.5);
        %\draw[black,very thick] (1,1) rectangle (4.5,4.5);
        %\draw[black,very thick] (0,0) rectangle (6.5,6.5);
%        \draw[step=5mm, black] (0,0) grid (6.5,6.5);
        \fill[black, fill opacity = 0.1] (-1,0) rectangle (4.5,4.5);
        \draw[-latex,very thick](4.5, 0)node[right, near end]{} to (4.5,4.5);
		%\draw[-latex,thick](4.5, 0)node[right, near end]{Space $\mathbf{s}$}
        %         to[out=180,in=180] (0,0);
    \end{scope}

    \begin{scope}[
    	yshift=-3*\gridSepCm,every node/.append style={
    	    yslant=0,xslant=0},yslant=0,xslant=-0.75]
        \draw[step=5mm, \gridColor] (0,0) grid (4.5,4.5);
        \fill[red,fill opacity=\opacLC] (0.5,0.5) rectangle (4,4);
        %\fill[green,fill opacity=\opacLC] (2,2) rectangle (2.5,2.5);
        \draw[black,very thick] (0.5,0.5) rectangle (4,4);
        \draw[step=5mm, \gridColor, opacity = \opacGrid] (0.5,0.5) grid (4,4);
        \draw[black, very thick, opacity = \opacRectBoundary] (0,0) rectangle (4.5,4.5);
        \fill[red, fill opacity = \opacBackground]  (0,0) rectangle (4.5,4.5);
    \end{scope}

    \begin{scope}[
    	yshift=-2*\gridSepCm,every node/.append style={
    	    yslant=0.5,xslant=-1},yslant=0,xslant=-0.75]
    	\draw[step=5mm, \gridColor] (0,0) grid (4.5,4.5);
        \fill[red,fill opacity=\opacLC] (1,1) rectangle (3.5,3.5);
        %\fill[green,fill opacity=\opacLC] (2,2) rectangle (2.5,2.5);
        \draw[black,very thick] (1,1) rectangle (3.5,3.5);
        \draw[step=5mm, \gridColor, opacity = \opacGrid] (1,1) rectangle (3.5,3.5);
        \draw[black,very thick, opacity = \opacRectBoundary] (0,0) rectangle (4.5,4.5);
        \fill[red, fill opacity = \opacBackground]  (0,0) rectangle (4.5,4.5);
    \end{scope}

    \begin{scope}[
    	yshift=-\gridSepCm,every node/.append style={
    	    yslant=0.5,xslant=-1},yslant=0,xslant=-0.75]
        \draw[step=5mm, \gridColor] (0,0) grid (4.5,4.5);
        \fill[red,fill opacity=\opacLC] (1.5,1.5) rectangle (3,3);
        %\fill[green,fill opacity=\opacLC] (2,2) rectangle (2.5,2.5);
        \draw[black,very thick] (1.5,1.5) rectangle (3,3);
        \draw[step=5mm, \gridColor, opacity = \opacGrid] (1.5,1.5) rectangle (3,3);
        \draw[black,very thick, opacity = \opacRectBoundary] (0,0) rectangle (4.5,4.5);
        \fill[red, fill opacity = \opacBackground]  (0,0) rectangle (4.5,4.5);
    \end{scope}

    \begin{scope}[
    	yshift=0,every node/.append style={
    	    yslant=0.5,xslant=-1},yslant=0,xslant=-0.75]
        \draw[step=5mm, \gridColor] (0,0) grid (4.5,4.5);
        \fill[green,fill opacity=\opacLC] (2,2) rectangle (2.5,2.5);
        \draw[black,very thick] (2,2) rectangle (2.5,2.5);
        \draw[step=5mm, \gridColor, opacity = \opacGrid] (2,2) rectangle (2.5,2.5);
        \draw[black,very thick, opacity = \opacRectBoundary] (0,0) rectangle (4.5,4.5);
        \fill[green, fill opacity = \opacBackground]  (0,0) rectangle (4.5,4.5);
    \end{scope}

    \begin{scope}[
    	yshift=\gridSepCm,every node/.append style={
    	    yslant=0.5,xslant=-1},yslant=0,xslant=-0.75]
    	
        \draw[step=5mm, \gridColor] (0,0) grid (4.5,4.5);
        \fill[blue,fill opacity=\opacLC] (1.5,1.5) rectangle (3,3);
        %\fill[green,fill opacity=\opacLC] (2,2) rectangle (2.5,2.5);
        \draw[black,very thick] (1.5,1.5) rectangle (3,3);
        \draw[step=5mm, \gridColor, opacity = \opacGrid] (1.5,1.5) rectangle (3,3);
        \draw[black,very thick, opacity = \opacRectBoundary] (0,0) rectangle (4.5,4.5);
        \fill[blue, fill opacity = \opacBackground] (0,0) rectangle (4.5,4.5);
    \end{scope}

    \begin{scope}[
    	yshift=2*\gridSepCm,every node/.append style={
    	    yslant=0.5,xslant=-1},yslant=0,xslant=-0.75]
        \draw[step=5mm, \gridColor] (0,0) grid (4.5,4.5);
        \fill[blue,fill opacity=\opacLC] (1,1) rectangle (3.5,3.5);
        %\fill[green,fill opacity=\opacLC] (2,2) rectangle (2.5,2.5);
        \draw[black,very thick] (1,1) rectangle (3.5,3.5);
        \draw[step=5mm, \gridColor, opacity = \opacGrid] (1,1) rectangle (3.5,3.5);
        \draw[black,very thick, opacity = \opacRectBoundary] (0,0) rectangle (4.5,4.5);
        \fill[blue, fill opacity = \opacBackground] (0,0) rectangle (4.5,4.5);
    \end{scope}

	%\draw[color=black,->, very thick] (0,0,0) -- (10,0,-2);
	\draw[color=black,->, very thick] (4.5,-3.5*\gridSep,0) -- node [sloped, below, near end, rotate =-90] {Time $t$}  (4.5,3*\gridSep,0);
	\draw[color=black,->, very thick] (4.5,-3.5*\gridSep,0) -- node [sloped, right, near end, rotate = 37] {Space $\mathbf{r}$}  (-2,-3.5*\gridSep,0);
%	\draw[color=black,->, very thick] (4.5,-3.5*\gridSep,0) -- node [sloped, right, near end, rotate = 37] {Space $\mathbf{r}$}  (1,-3*\gridSep,0);

	%\foreach \x in{1, 2, 3}
    %\draw[color=blue,->, thick] (4.5 - 0.5*\x, -3*\gridSep, \x) -- (4.5 - \x, 0, \gridSep*\x);
    
    %\draw[color=red,->, thick] (4, -3*\gridSep, 1) -- (2, 0, 2);
    %\draw[color=red,->, thick] (1, -2*\gridSep, 4) -- (1, 3*\gridSep, 4);

    %putting arrows and labels:
    \draw[-latex,thick, green] (7,0) node[below]{present $t=t_0$} to[out=180,in=90] (4.5,0);

    %\draw[-latex,thick](5.8,-.3)node[right]{$\mathsf{Comp.\ G.}$}
    %    to[out=180,in=90] (3\opacLC,-1);

    %\draw[-latex,thick] (6,-2) node[right]{past light cone}
    %     to[out=180,in=180] (4,-2);

	\draw (8,-7) node [text width=2.5cm, text centered, red]{past $t <t_0$ \\ \textbf{PLC}};
	\draw (-5, 9) node [text width=2.5cm, text centered, blue]{future $t>t_0$ \\ \textbf{FLC}};
    %drawing points on grid's conrners.

\end{tikzpicture}

%% file: NonparametricForecastingUsingLightCones_supplementary_material_input.tex
\begin{center}
\Large Supplementary Material for ``\ThesisTitle''
\end{center}

\appendix

\section{Predictive States: Details on Methodology, Implementation, and Algorithms}

\begin{figure}[!t]
  \begin{center}
    \fbox{
      \begin{minipage}{0.99\textwidth}
        \begin{enumerate}
        \item Collect the PLC and FLC configurations, $\field{\ell^{-}}{r}{t}$
          and $\field{\ell^{+}}{r}{t}$, for each $\field{}{r}{t}$ in the
          observed field $\field{X}{r}{1}, \ldots, \field{X}{r}{T}$.
        \item \label{step:cluster_no_cluster} To cluster or not to cluster:
          \begin{enumerate}
          \item \label{step:no_clustering} Assign each point to its own
            cluster. Only for small $N$ this is computationally feasible.
          \item \label{step:clustering} Perform an initial clustering (e.g.\
            (e.g., K-means++ \citepapp{k-means-plus-plus}) in the PLC
            configuration space (Section \ref{sec:partition}).
          \end{enumerate}
        \item \label{step:merging} For each pair of clusters, test whether the
          estimated conditional FLC distributions are significantly different,
          at some fixed level $\alpha$ (Section \ref{sec:testing_equality}).
          If not, merge them and go on.  Stop when no more merges are possible.
        \item Treat the remaining clusters as predictive states, and estimate
          the conditional distributions over FLC configurations.
        \item Return the partition of PLC configurations into predictive
          states, and the associated predictive distributions.
        \end{enumerate}
      \end{minipage}
    }
  \end{center}
  \caption{\label{fig:PS_Estimation_Overview} Estimating predictive states from
    continuous-valued data: in \ref{step:no_clustering} conditional
    distributions are tested for each $\ell^{-}_i$, $i = 1, \ldots, N$, using a
    $\delta$-neighborhood (or $k$ nearest neighbors) of $\ell^{-}_i$ (see
    Section \ref{sec:getting_samples} for details); \ref{step:clustering} uses
    an initial clustering to reduce complexity of the testing problem from
    $\mathcal{O}(N^2)$ to $\mathcal{O}(K^2)$ (see also Section
    \ref{sec:partition}).}
\end{figure}

\begin{algorithm}[!t] 
  \caption{\label{alg:test_predictive_distributions} Test equality of
    conditional predictive FLC distributions $\Prob{L^{+} \mid \text
      {clusterID} = k}$} % give the algorithm a caption
  \begin{algorithmic} % enter the algorithmic environment
    \REQUIRE

    \STATE
    \begin{tabular}{|lcl|}
      \hline
      Data: & & \\
      $\mathbf{F} = \lbrace \ell^{+}_i \rbrace_{i=1}^{N} \in \R^{N \times n_f}$ & $\ldots$ & array with FLCs \\
      clusterID & $\ldots$ & labels of the PLC partitioning (step \ref{step:no_clustering} or \ref{step:clustering} in Fig.\ \ref{fig:PS_Estimation_Overview})\\
      \hline
      Parameters: &   & \\
      $\alpha \in [0, 1]$ & $\ldots$ & significance level $\alpha$ for testing $H_0: \Prob{L^{+} \mid \ell^{-}_i} = \Prob{ L^{+} \mid \ell^{-}_j} $\\
      \hline
    \end{tabular}
    \ENSURE

    \STATE $k_{\max} = \max$ clusterID

\FOR{$k = 1, \ldots, k_{\max}$}
\STATE fetch FLC samples given partition $P_k$: $\mathbf{F}_k = \lbrace \ell^{+}_i \rbrace_{ \lbrace i \mid clusterID[i] == k \rbrace }$
\STATE j = k
\STATE lasttested = 0
\STATE pvalue = 1
%\STATE alldifferent = FALSE
\WHILE{pvalue $> \alpha$ or $j \leq k_{\max}$}
\STATE j = j+ 1

\STATE lasttested = j
\STATE fetch FLC samples given partition $P_j$: $\mathbf{F}_j = \lbrace \ell^{+}_i \rbrace_{ \lbrace i \mid clusterID[i] == j \rbrace }$
\STATE pvalue $\leftarrow test(\Prob{ L^{+} \mid P_k} = \Prob{ L^{+} \mid P_j} \mid \mathbf{F}_k, \mathbf{F}_j)$
\IF{pvalue $< \alpha$}
\STATE merge cluster $j$ with cluster $k$: clusterID[clusterID == j] = k
\ENDIF
\ENDWHILE
\ENDFOR

After no merging is possible clusterID contains the labels of the predictive states.

\RETURN clusterID 
\end{algorithmic}

\end{algorithm}

We partition the observed PLCs $\lbrace \ell_{i}^{-}
\rbrace_{i=1}^{N} \subset \R^{n_p}$ into $K = K(\delta)$ disjoint groups
$\lbrace P_k \rbrace_{k=1}^{K}$, choosing the number of groups so that all have
diameters less than $\delta$.  This choice of $K(\delta)$ guarantees
(Assumption \ref{ass:continuous_dynamics}) that all $\ell^{-} \in P_k$ have
predictive distributions that are at most $\rho$ apart.  Thus all PLCs within a
group $P_k$ are (nearly) equivalent by Definition
\ref{def:equivalent_configurations}.  This in turn means we only need to
compare predictive distributions between clusters. 

\subsection{Lebesgue Smoothing}
\label{sec:Lebesgue_smoothing}

In a standard kernel regression approach one would compute a similarity measure
on PLCs $\ell^{-}_i$ and then use a weighted mean of FLCs $\ell_i^{+}$ to get a
point prediction of the future cone, i.e.,
\begin{equation}
  \label{eq:Riemann_Kernel_Prediction}
  \field{\widehat{L^{+}}}{r}{t} = \sum_{\field{}{q}{\tau}}{w^{-}\left( \vecS{q}, \tau; \field{}{r}{t} \right) \field{\ell^{+}}{q}{\tau}} \text{ for all } \tau < t,
\end{equation}
where $w^{-}\left( \vecS{q}, \tau; \field{}{r}{t} \right) \propto K^{-}\left(
  \vnorm{\field{\ell^{-}}{r}{t} - \field{\ell^{-}}{q}{\tau}} \right)$ are
normalized weights determined by a kernel $K^{-}\left( \cdot \right)$ in the
PLC configuration space. For example, a Gaussian kernel $K_h^{-}\left(
  \field{\ell^{-}}{r}{t}, \field{\ell^{-}}{q}{\tau} \right) = \exp(-
\frac{1}{h} \vnorm{ \field{\ell^{-}}{r}{t}- \field{\ell^{-}}{q}{\tau}}_2^2 )$
with squared Euclidean distance and bandwidth $h$.

Since $\field{}{q}{\tau}$ ranges over the entire space-time, $\vecS{q} \in
\vecS{S}$, $\tau = 1, \ldots, t-1$, computing this many similarities $\lbrace
s_{\field{}{q}{\tau}, \field{}{r}{t} } \rbrace$ becomes very time consuming.  A
typical $10$-second video might have $N=3 \cdot{10}^7$ space-time
points.\footnote{$25$ frames per second and $300 \times 400$ pixels.}  To
evaluate \eqref{eq:Riemann_Kernel_Prediction} needs $3 \cdot 10^7$ similarities
in $n_p$-dimensional space --- and this just to predict one FLC.  If $N$ is
large, then predictive state estimation is a necessary pre-step before making
predictions.

\begin{comment}
Our approach differs in two important ways. First, we assume a discrete
predictive state space which is sufficient to predict the future. Thus once we
have estimated the predictive states $\epsilon_1, \ldots, \epsilon_m$, 
%and corresponding $\widehat{p}_1, \ldots, \widehat{p}_m$,
 we can predict the field
at any $\field{}{r}{t}$ using the average (or mode) of $p_{\widehat{S}}$, where
$\widehat{S} = \field{\widehat{S}}{r}{t}$ is the estimated predictive state at
$\field{}{r}{t}$,
\begin{equation}
\label{eq:state_space_FLC_pred}
\field{\widehat{L^{+}}}{r}{t} = \E_{\widehat{p}_{\field{\widehat{S}}{r}{t}}} \left(  \field{L^{+}}{r}{t} \right).
\end{equation}
\end{comment}

Our approach differs in two important ways. First, we assume a discrete
predictive state space which is sufficient to predict the future. Thus once we
have estimated the predictive states $\epsilon_1, \ldots, \epsilon_m$, we can predict the field
at any $\field{}{r}{t}$ using the average (or mode) of the estimated predictive state at
$\field{}{r}{t}$,
\begin{equation}
\label{eq:state_space_FLC_pred}
\field{\widehat{L^{+}}}{r}{t} = \E_{\widehat{\epsilon}\left(\field{\ell^{-}}{r}{t} \right) } \left(  \field{L^{+}}{r}{t} \right).
\end{equation}

Second, we learn a new geometry on the PLC space by defining closeness in the
FLC distribution space, rather than in the PLC configuration space. Thus a
natural continuous state space extension of \eqref{eq:state_space_FLC_pred} is
a Kernel regression with weights that depend on the similarity in the output
rather than the input space, i.e.\
\begin{equation}
  \label{eq:Lebesgue_Kernel_Prediction}
  \field{\widehat{L^{+}}}{r}{t} = \sum_{\field{}{q}{\tau}}{ w^{+}\left( \vecS{q}, \tau; \field{}{r}{t}  \right) \field{\ell^{+}}{q}{\tau}} \text{ for all } \tau < t,
\end{equation}
where the normalized weights $w^{+}\left( \vecS{q}, \tau; \field{}{r}{t}
\right) \propto K^{+}\left( \vnorm{ \Prob{\field{\ell^{+}}{r}{t}}
    - \Prob{\field{\ell^{+}}{q}{\tau}}} \right)$ are based on a
Kernel $K^{+}\left( \cdot \right)$ in the FLC distribution space.

One can generalize \eqref{eq:Lebesgue_Kernel_Prediction} to the classic
non-parametric regression setting $y = m(x) + u$ and define a new Kernel
regression estimator as
\begin{equation}
  \label{eq:Lebesgue_smoothing}
  \widehat{m}^{(L)}(x) = \sum_{i=1}^{n}{\frac{K_y\left( \widehat{m}^{(R)}(x_i) - m(x) \right)}{h_y} y_i} ~ ,
\end{equation}
where $\widehat{m}^{(R)}(\cdot)$ serves as a pilot estimate; for example the
classic kernel regression smoother
\begin{equation}
  \label{eq:Riemann_smoothing}
  \widehat{m}^{(R)}(x) = \sum_{i=1}^{n}{\frac{K_x\left( x_i - x \right)}{h_x} y_i} ~ .
\end{equation}
As we average over nearby predictions rather than nearby inputs, we may call
\eqref{eq:Lebesgue_smoothing} ``Lebesgue smoothing'', in contrast to the
``Riemann'' smoothing of \eqref{eq:Riemann_smoothing}.  If $N$ is small, then
we can forecast with \eqref{eq:Lebesgue_Kernel_Prediction} forecast without
estimating predictive states.  However, here we focus on predictive-state
recovery, and leave Lebesgue smoothed LICORS to future work.

\paragraph{Further performance enhancements for testing}

While it is better to do $\mathcal{O}(K^2)$ high-dimensional tests than
$\mathcal{O}(N^2)$, it would be better still to speed up each test.  Since two
distributions are the same only if their moments are, we can start by testing
simply for equality of means, which is fast and powerful, and do a full
distributional test only if we cannot reject on that basis.  For multivariate
mean tests we can use the Hotelling test \citepapp{ABW98} and its randomized
generalization \citepapp{Lopes-et-al-randomized-Hotelling-test}.  Yet another
strategy to reduce the number of costly high-dimensional, non-parametric tests
is to test various functions $f(\cdot)$ of the samples.  If the distributions
of $\mathbf{F}_{k_i}(\delta)$ and $\mathbf{F}_{k_j}(\delta)$ are the same, then
also $\Prob{f\left( \mathbf{F}_{k_i}(\delta) \right)} = \Prob{f\left(
    \mathbf{F}_{k_j}(\delta) \right)}$ for any measurable $f$. Particularly, we
can apply random projections \citepapp{Lopes-et-al-randomized-Hotelling-test} to
$\mathbf{F}_{k_i}$ to go from the high-dimensional $\R^{n_f}$ down to the
one-dimensional $\R$, followed by a Kolmogorov-Smirnov test.  Only if these
tests can not reject equality for several projections, one uses fully
non-parametric tests.

\section{The Simulation and the Forecasting Competition}

\subsection{Details of Competing Methods}
\label{sec:competing-methods}

%The local VAR models were fit using a simple linear regression using
%the \texttt{lm} function with time-lagged variables.
The local VAR models were fit with Lasso regularization
\citepapp{BickelSong11_LargeVAR}, as implemented in the \texttt{fastVAR}
package \citepapp{fastVAR}. We also tried un-regularized VAR models,
but they performed even worse.

\begin{figure}[!t]
  \begin{center}
    \fbox{
      \begin{minipage}{0.99\textwidth}
        \begin{enumerate}
        \item \label{step:split_half} Split dataset at its middle in time:
          $\mathcal{D}_1 = \lbrace \field{X}{r}{t} \rbrace_{t=1}^{T/2}$ and
          $\mathcal{D}_2 = \lbrace \field{X}{r}{t} \rbrace_{t=T/2 + 1}^{T}$
        \item \label{step:train_1st} For each combination of control settings,
          do:
          \begin{enumerate}
          \item Training: estimate predictive states from $\mathcal{D}_1$
          \item Test-set prediction: find predictive state of each PLC $\in
            \mathcal{D}_2$ and predict its FLC $\in \mathcal{D}_2$.
          \item Error: compare to the observed FLCs $\in \mathcal{D}_2$ and
            compute the loss.
          \end{enumerate}
        \item Choose the control settings with the smallest test-set loss.
        \end{enumerate}
      \end{minipage}
    }
  \end{center}
  \caption{\label{fig:CV_Overview} Cross-validation to choose control settings
    given data $\lbrace \field{X}{r}{t} \rbrace_{t=1}^{T}$.}
\end{figure}

\subsection{Excess Risk, Test Size, and Number of Estimated States}
\label{sec:excess-risk}

Figs.\ \ref{fig:ContCA_h_p=2_sd1_trunc4_outofsampleCV_MSE_nclusters200_params}
and \ref{fig:ContCA_h_p=2_sd1_trunc4_outofsampleCV_MSE_knn50_params} show the
expected relationship between $\alpha$ and the number of predictive states
recovered $\widehat{m}$: smaller $\alpha$ leads to more merging, and fewer
states.  Here the true number of states $m=7$, but both pre-clustering and
direct estimation give much higher $\widehat{m}$.  Thus for LICORS, optimal
forecasting pushes for more states and lower approximation error, rather than
fewer states and lower estimation error.  We can check this explanation by
considering the ratio
\begin{equation}
\label{eq:MSE_ratio}
\text{excess risk} := \frac{\text{MSE(sample $i+1$)} \text{ using } (h_p, \alpha)_{i, CV_i}}{\text{MSE(sample $i+1$)} \text{ using } (h_p, \alpha)_{i+1, \min}} \geq 1.
\end{equation}
Recall that $(h_p, \alpha)_{i, CV_i}$ is chosen using only sample $i$, while
$(h_p, \alpha)_{i+1, \min}$ is the minimizing pair after having evaluated the
MSE on sample $i+1$. The best that any data-driven procedure could do would be
to guess $(h_p, \alpha)_{i+1, \min}$ from sample $i$, so the excess risk is
$\geq 1$, with equality only if CV picked the optimal control settings.  The
scatter-plots show that our CV procedure has an excess risk on the order of
$10^{-2}$ compared to the oracle pair.  Hence, even though $\widehat{m}$ is
substantially larger than $m$, the difference is practically irrelevant for
predictions.

\begin{figure}[!t]
  \centering 
  \subfloat[]
  {\includegraphics[width=.4\textwidth]{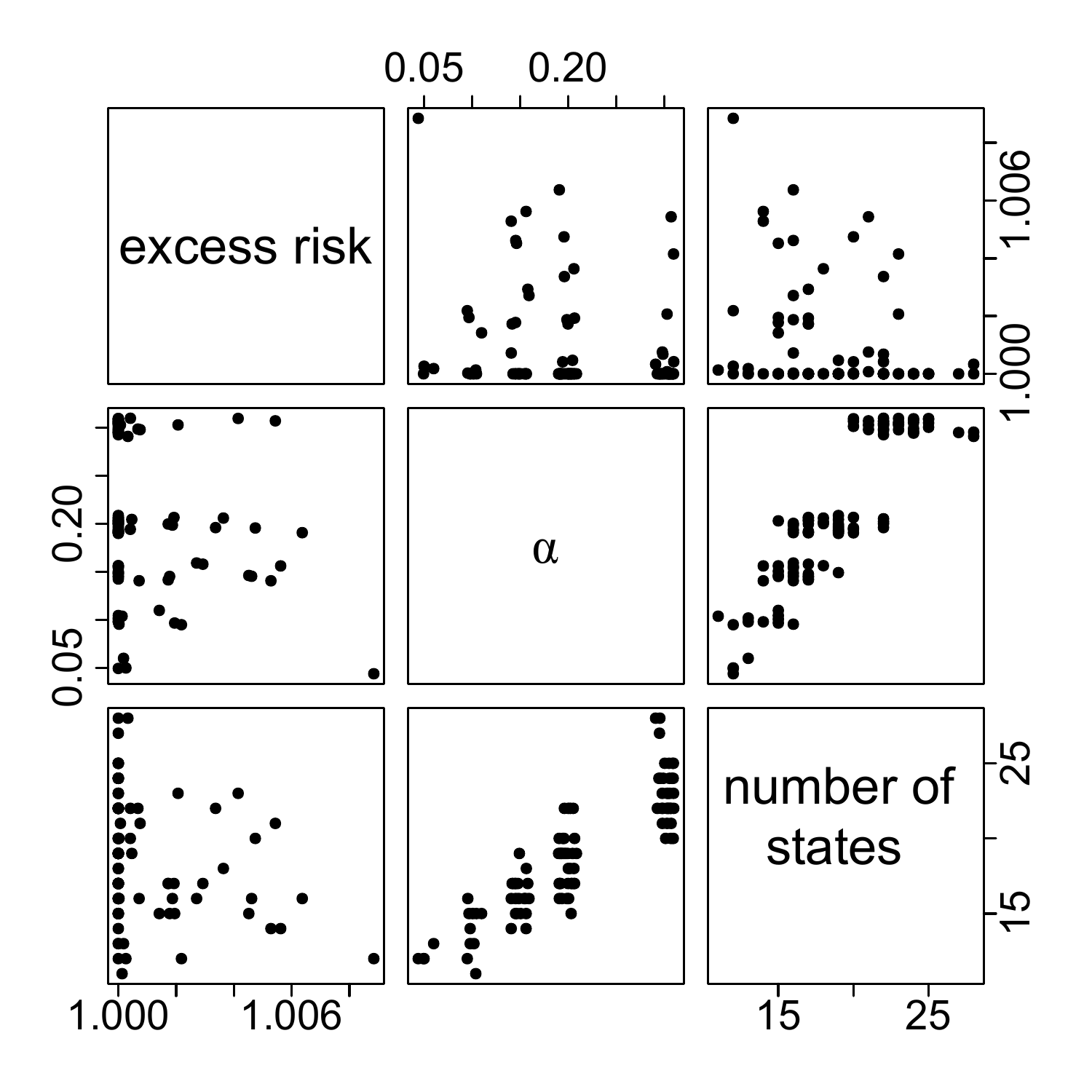}\label{fig:ContCA_h_p=2_sd1_trunc4_outofsampleCV_MSE_nclusters200_params}
  } \hspace{0.02\textwidth} \subfloat[]
  {
  \includegraphics[width=.4\textwidth]{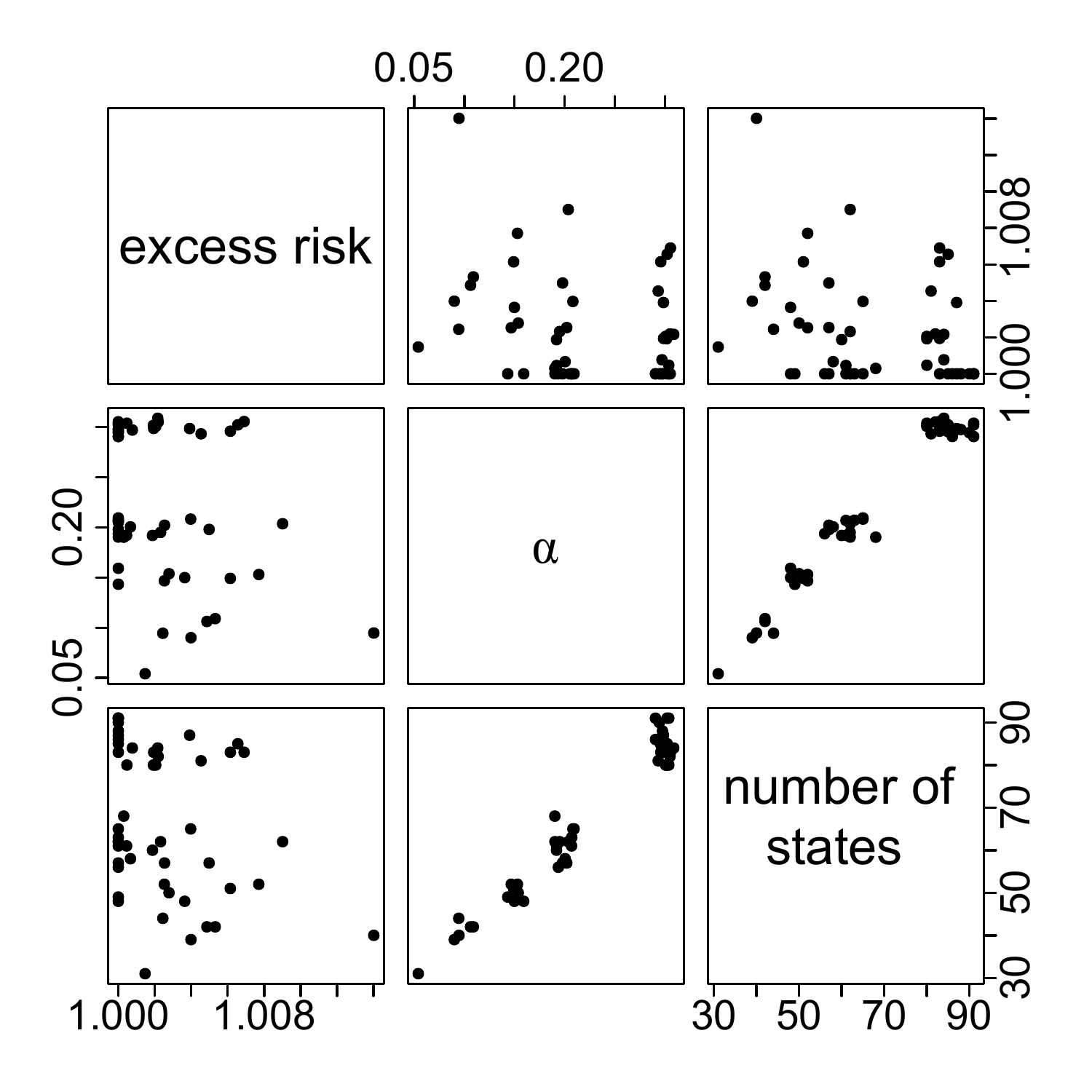}\label{fig:ContCA_h_p=2_sd1_trunc4_outofsampleCV_MSE_knn50_params}
  }

  \caption{\label{fig:CV_Excess_risk_vs_states} Relations between excess risk,
    test size, and the number of reconstructed states for LICORS: (a) selected
    $\alpha$, number of estimated states, and excess risk (Eq.\
    \eqref{eq:MSE_ratio}) for pre-clustered LICORS; (b) the same for
    direct-estimation LICORS. $\alpha$ values in are jittered.}
\end{figure}

%%%%%%%%%%%%%%%%%
% [[CRS: Again, the numbers around the bars in (a) are too small and too hard
% to interpret; a table in the SI.]]
%%%%%%%%%%%%%%%%%%

\subsection{Discussion of the Simulations}

The simulations showed that LICORS outperforms standard forecasting techniques
by a large margin, even though it presumes very little about the data
source. Especially note that the out-of-sample MSE in Fig.\
\ref{fig:CV_LICOR_comparisons} is still much lower than the best parametric
in-sample MSE in Fig.\ \ref{fig:ERM_AR_VAR_LICOR_nclusters_knn} --- even though
it uses only half the sample size.  The good performance of the CV procedure
(Fig.\ \ref{fig:CV_Overview}) suggests that using it to find control settings
in applications will avoid over-fitting.

In real applications $N$ would typically on the order of millions (rather than
merely $2\times{10}^4$), making pre-clustering essential computationally --- at
least until $\bigO{N^2}$ comparisons for millions of data points become
tractable.  Pre-clustering usually leads to a performance loss as it hides fine
structures in the predictive distribution space (see also the remark 
below Lemma \ref{lem:F_i_is_sample_from_p_i}). However, the in-sample
and out-of-sample MSE comparison showed that this performance loss is small
compared to the gain over standard parametric methods, and further attenuated
with CV.

\section{Proofs}
\label{sec:proofs}

\begin{proof}[Proof of Lemma \ref{lem:conditional_independent_FLC_given_states}]

  \begin{eqnarray}
    \lefteqn{\Prob{\field{L^+}{r}{t}, \field{L^+}{u}{s} \mid
      \field{S}{r}{t}, \field{S}{u}{s}}} & &\\
   & = & \Prob{\field{L^+}{r}{t} \mid \field{L^+}{u}{s}, \field{S}{r}{t}, \field{S}{u}{s}}  \Prob{\field{L^+}{u}{s} \mid \field{S}{r}{t}, \field{S}{u}{s}} \\
   & = & \Prob{\field{L^+}{r}{t} \mid \field{L^+}{u}{s},  \field{S}{r}{t}, \field{S}{u}{s}}  \Prob{\field{L^+}{u}{s} \mid \field{S}{u}{s}} \\
    & = & \Prob{\field{L^+}{r}{t} \mid \field{S}{r}{t}} \Prob{\field{L^+}{u}{s} \mid \field{S}{u}{s}} ~ ,
  \end{eqnarray}
  The first equality is simple conditioning, the second equality holds since
  given the predictive state at $\field{}{u}{s}$ the distribution of $L^+$
  is independent of the predictive state at another $\field{}{r}{t}$, and the
  last equality holds for the same reason as the second plus the non-overlap of
  the FLCs at $\field{}{r}{t}$ and $\field{}{u}{s}$.
\end{proof}

\begin{proof}[Proof of Corollary
  \ref{cor:univariate_FLCs_conditional_independent}]
  The FLC of $\field{}{r}{t}$ with $h_f = 0$ is just the single point
  $\field{X}{r}{t}$. Since two univariate FLCs cannot overlap unless they are
  equal, the result follows immediately from Lemma
  \ref{lem:conditional_independent_FLC_given_states}.
\end{proof}

\begin{proof}[Proof of Lemma \ref{lem:F_i_is_sample_from_p_i}]
  By contradiction.  Assume that $\ell_{j}^{-}$ and $\ell^{-}_{k}$, with $j, k
  \in I_i(\delta)$, have different predictive states, without loss of
  generality $\epsilon_1$ and $\epsilon_2$.  By Assumption
  \ref{ass:distance_between_states}, then, $\KL{\epsilon_1}{\epsilon_2}$ and
  $\KL{\epsilon_2}{\epsilon_1}$ are both at least $d_{\min}$.  By the
  definition of $I_i(\delta)$, $\vnorm{\ell_j^{-} - \ell_k^{-}} < 2\delta$.  By
  Assumption \ref{ass:continuous_dynamics}, then, $\KL{\epsilon_1}{\epsilon_2}$
  and $\KL{\epsilon_2}{\epsilon_1}$ are both at most $\rho(2\delta)$.  But by
  making $\delta$ sufficiently small, $\rho(2\delta)$ can be made as small as
  desired, and in particular can be made less than $d_{\min}$.  This is a
  contradiction, so all the past cone configurations in $I_i(\delta)$ must be
  predictively equivalent.
\end{proof}

\begin{proof}[Proof of Corollary \ref{cor:likelihood_factorizes}]
  Immediate from combining Lemmas \ref{lem:F_i_is_sample_from_p_i} and
  \ref{lem:conditional_independent_FLC_given_states}.
\end{proof}

%\begin{proof}[Proof of Corollary \ref{cor:finite_states_oracle}]
%  Follows immediately from Theorem \ref{thm:asymptotic_perfect}, as
%  $\log{m(N)}$ is a constant.
%\end{proof}

\begin{proof}[Proof of Theorem \ref{thm:consistent_unknown}]
  Before going into the formal proof, we make an observation regarding
  non-parametric two-sample tests.  Most of these, to have good operating
  characteristics, require independent samples.  Since we will
be applying the tests to $\mathbf{F}_i(\delta)$ and $\mathbf{F}_j(\delta)$,
\begin{properties}[Pairwise independent samples]
    \label{ass:independent_samples}
    If
    \begin{equation}
      \label{eq:no_intersection}
      I_i(\delta) \cap I_j(\delta) = \emptyset.
    \end{equation}
    then the samples $\mathbf{F}_i(\delta)$ are independent of
    $\mathbf{F}_j(\delta)$, $j \neq i$ (see
    \eqref{eq:indices_PLC_neighborhood}).
  \end{properties}

  Let $\Delta_{ij} := \vnorm{ \ell_i^{-} - \ell_j^{-} }$. If $\Delta_{ij} > 2
  \delta$, then \eqref{eq:no_intersection} is satisfied. If $\Delta_{ij} < 2
  \delta$, then a sample in $\mathbf{F}_i(\delta)$ might also appear in
  $\mathbf{F}_j(\delta)$ and therefore violate the independence assumption for
  two sample tests.

  For these rare cases redefine the index set $I_i(\delta)$ and $I_j(\delta)$
  such that \eqref{eq:no_intersection} holds. We can achieve this by excluding
  the intersection, split it in half ( $\pm 1$ sample), and then
  re-assign these halves to each index set.  For all pairs $i \neq j$ 
  determine $I_i(\delta) \cap I_j(\delta) =: I_{i \cap j}(\delta)$. Then let
  \begin{align}
    \label{eq:redefine_indices_i}
    I_i & := I_i \setminus I_{i \cap j} \cup \lbrace i_1, \ldots, i_{\card{I_{i \cap j}}/2}\mid i_k \in I_{i \cap j} \rbrace\\
    \label{eq:redefine_indices_j}
    \text{ and } I_j & := I_j \setminus I_{i \cap j} \cup \lbrace i_{\card{I_{i
          \cap j}}/2}, \ldots, i_{\card{I_{i \cap j}}} \mid i_k \in I_{i \cap
      j} \rbrace.
  \end{align}

  If $I_{i \cap j} = \emptyset$,
  \eqref{eq:redefine_indices_i}--\eqref{eq:redefine_indices_j} does not change
  the index set; if $I_{i \cap j} \neq \emptyset$, then
  \eqref{eq:redefine_indices_i}--\eqref{eq:redefine_indices_j} guarantees an
  empty intersection.

  The proof of consistency relies crucially on a growing index set $I_i$. The
  re-definition in \eqref{eq:redefine_indices_i}--\eqref{eq:redefine_indices_j}
  does not change the rate at which $S_i(N, \delta)$ grows, because in the
  worst case (for very close PLCs) it just divides $s_{i}(N, \delta)$ and
  $s_{j}(N, \delta)$ in half.

  \paragraph{Proof:}
  We first bound the error for each row $\widehat{\mathbf{A}}_i$, and then use
  a union bound for the probability of error for $\widehat{\mathbf{A}}$.

  \paragraph{Bound error per row}
  For each row $T_{n,m}$ tests $H_0: \ell_i \sim \ell_j$, $j > i$ (due to
  symmetry the cases $j<i$ have already been tested before) based on the sample
  $\mathbf{F}_i(\delta) \sim \epsilon_i$ and $\mathbf{F}_j(\delta) \sim
  \epsilon_j$. The worst-case distance $d$ for the non-parametric test in
  Assumption \ref{ass:existence_consistent} is $d = d_{\min}$. For simplicity
  consider the first row: here we have to make $N-1$ tests, of which $N_1 - 1$
  should correctly accept, and $N - N_1$ should correctly reject equality of
  distributions.

  \begin{eqnarray}
    \Prob{\widehat{\mathbf{A}}_j \neq \mathbf{A}_j} 
    & \leq & (N_j - 1) \Prob{ \left( \text {type I} \right) + (N - N_j)} \Prob{\text {type II}} \\
    & \leq & (N_j - 1) \alpha + (N - N_j) \beta\left( \alpha, S_{\min}(N, \delta), S_{\min}(N, \delta)  \right) \\
    & \leq & N_j \alpha + (N - N_j) \beta\left( \alpha, S_{\min}(N, \delta), S_{\min}(N, \delta) \right)
  \end{eqnarray}
  since the worst case, for type II error, is that both samples are as small as
  possible.

  \paragraph{Bound error for entire matrix}
  The probability of error for the entire predictive state clustering can again
  be bounded using the union bound:

  \begin{align}
    \Prob{ \widehat{\mathbf{A}} \neq \mathbf{A}} & = \Prob{\bigcup_{j = 1}^{N}{ \lbrace \widehat{\mathbf{A}}_j \neq \mathbf{A}_j \rbrace}} \\
    & \leq \sum_{j=1}^{N}{\Prob{\widehat{\mathbf{A}}_j \neq \mathbf{A}_j}} \\
    & \leq N \left( N_{\max} \alpha + (N - N_{\min}) \beta\left( \alpha, S_{\min}(N, \delta), S_{\min}(N, \delta) \right) \right) \\
    & = N N_{\max} \alpha + (N^2 - N N_{\min}) \beta\left( \alpha, S_{\min}(N,
      \delta), S_{\min}(N, \delta) \right),
  \end{align}
  where $N_{\max} = \max_{j} N_j$ is the number of light cones in the largest
  predictive state.

  Under Assumption \ref{ass:existence_consistent}, $\alpha$ and $\beta$ are
  both $\littleO{N N_{\max}}$, so the over-all arrow probability tends to zero.
\end{proof}

\bibliographystyleapp{plainnat}
\bibliographyapp{../../../bib/PhD_thesis,../../../bib/locusts}